\documentclass[sn-mathphys,Numbered]{sn-jnl}

\usepackage{amssymb,amsmath,amsfonts,mathrsfs,amsthm}
\usepackage[utf8]{inputenc}
\usepackage[T1]{fontenc}
\usepackage{caption}
\usepackage[english]{babel}
\usepackage{graphicx}
\usepackage[title]{appendix}
\usepackage{color}
\usepackage{doi}
\usepackage{commath}
\usepackage{array}
\usepackage[normalem]{ulem}
\usepackage{multirow}%
\usepackage{mathrsfs}%
\usepackage{xcolor}%
\usepackage{textcomp}%
\usepackage{manyfoot}%
\usepackage{booktabs}%

\theoremstyle{thmstyleone}%
\newtheorem{theorem}{Theorem}
\newtheorem{lemma}{Lemma}%
\newtheorem{corollary}{Corollary}
\newtheorem{proposition}{Proposition}

\theoremstyle{thmstyletwo}%
\newtheorem{remark}{Remark}%

\theoremstyle{thmstylethree}%
\newtheorem{definition}{Definition}%
\newtheorem{problem}{Problem}

\DeclareMathOperator{\CC}{\mathbb{C}}

\DeclareMathOperator{\RR}{\mathbb{R}}

\DeclareMathOperator{\QQ}{\mathbb{Q}}

\DeclareMathOperator{\ZZ}{\mathbb{Z}}

\DeclareMathOperator{\BB}{\mathbb{B}}

\DeclareMathOperator{\GC}{\mathcal{G}}
\DeclareMathOperator{\RC}{\mathcal{R}}
\DeclareMathOperator{\BC}{\mathcal{B}}
\DeclareMathOperator{\NC}{\mathcal{N}}
\DeclareMathOperator{\VC}{\mathcal{V}}

\DeclareMathOperator{\EC}{\mathcal{E}}
\DeclareMathOperator{\PC}{\mathcal{P}}

\DeclareMathOperator{\HC}{\mathcal{H}}
\DeclareMathOperator{\MC}{\mathcal{M}}
\DeclareMathOperator{\AC}{\mathcal{A}}
\DeclareMathOperator{\FC}{\mathcal{F}}
\DeclareMathOperator{\KC}{\mathcal{K}}
\DeclareMathOperator{\UC}{\mathcal{U}}

\DeclareMathOperator{\JC}{\mathcal{J}}
\DeclareMathOperator{\IC}{\mathcal{I}}

\DeclareMathOperator{\QC}{\mathcal{Q}}

\DeclareMathOperator{\ES}{\mathscr{E}}

\DeclareMathOperator{\RS}{\mathscr{R}}

\DeclareMathOperator{\MS}{\mathscr{M}}

\DeclareMathOperator{\PS}{\mathscr{P}}

\DeclareMathOperator{\QS}{\mathscr{Q}}

\DeclareMathOperator{\fG}{\mathfrak{f}}
\DeclareMathOperator{\gG}{\mathfrak{g}}
\DeclareMathOperator{\hG}{\mathfrak{h}}

\DeclareMathOperator{\rank}{rank}
\DeclareMathOperator{\cone}{cone}

\DeclareMathOperator{\vertex}{vert}

\DeclareMathOperator{\size}{size}
\DeclareMathOperator{\poly}{poly}

\DeclareMathOperator{\inter}{int}

\DeclareMathOperator{\vol}{vol}

\DeclareMathOperator{\lmod}{\text{\it\quad modulo polyhedra with lines}}

\DeclareMathOperator{\toddp}{td}

\DeclareMathOperator{\disc}{disc}
\DeclareMathOperator{\herdisc}{herdisc}

\DeclareMathOperator{\detlb}{detlb}

\DeclareMathOperator{\BUnit}{\mathbf 1}
\DeclareMathOperator{\BZero}{\mathbf 0}
\DeclareMathOperator{\xB}{\mathbf{x}}

\DeclareMathOperator{\SNF}{\text{\rm SNF}}

\DeclareMathOperator{\rowSparse}{rs}
\DeclareMathOperator{\colSparse}{cs}
\DeclareMathOperator{\Sparse}{ts}
\DeclareMathOperator{\wrowSparse}{\overline{\rowSparse}}
\DeclareMathOperator{\wcolSparse}{\overline{\colSparse}}
\DeclareMathOperator{\wSparse}{\overline{\Sparse}}

\newcommand*{\intint}[2][1]{\{#1, \dots, #2\}}

\newcommand\restr[2]{{
  \left.\kern-\nulldelimiterspace 
  #1 
  \vphantom{\big|} 
  \right|_{#2} 
  }}

\DeclareMathOperator{\numv}{\mathfrak{n}}
\DeclareMathOperator{\nume}{\mathfrak{m}}
\DeclareMathOperator{\degv}{\mathfrak{d}}
\DeclareMathOperator{\dege}{\mathfrak{r}}
\DeclareMathOperator{\totn}{\gamma_{1,\infty}}

\begin{document}


\title[Faster Algorithms for Sparse ILP]{Faster Algorithms for Sparse ILP and Hypergraph Multi-Packing/Multi-Cover Problems
}

\author*[1]{\fnm{Dmitry} \sur{Gribanov}}\email{dimitry.gribanov@gmail.com}

\author[2]{\fnm{Ivan} \sur{Shumilov}}\email{ivan.a.shumilov@gmail.com}

\author[1]{\fnm{Dmitry} \sur{Malyshev}}\email{dsmalyshev@rambler.ru}

\author[3]{\fnm{Nikolai} \sur{Zolotykh}}\email{nikolai.zolotykh@itmm.unn.ru}

\affil*[1]{\orgdiv{Laboratory of Algorithms and Technologies for Network Analysis}, \orgname{HSE University}, \orgaddress{\street{136 Rodionova Ulitsa}, \city{Nizhny Novgorod}, \postcode{603093}, \country{Russian Federation}}}

\affil[2]{\orgname{Lobachevsky State University of Nizhny Novgorod}, \orgaddress{\street{23 Gagarina Avenue}, \city{Nizhny Novgorod}, \postcode{603950}, \country{Russian Federation}}}

\affil[3]{\orgname{Mathematics of Future Technologies Center, Lobachevsky State University of Nizhni Novgorod}, \orgaddress{\street{23 Gagarina Avenue}, \city{Nizhny Novgorod}, \postcode{603950}, \country{Russian Federation}}}
\abstract{
In our paper, we consider the following general problems: check feasibility, count the number of feasible solutions, find an optimal solution, and count the number of optimal solutions in $\PC \cap \ZZ^n$, assuming that $\PC$ is a polyhedron, defined by systems $A x \leq b$ or \mbox{$Ax = b,\, x \geq \BZero$} with a sparse matrix $A$. We develop algorithms for these problems that outperform state of the art ILP and counting algorithms on sparse instances with bounded elements.

We use known and new methods to develop new exponential algorithms for \emph{Edge/Vertex Multi-Packing/Multi-Cover Problems} on graphs and hypergraphs. This framework consists of many different problems, such as the \emph{Stable Multi-set}, \emph{Vertex Multi-cover}, \emph{Dominating Multi-set}, \emph{Set Multi-cover}, \emph{Multi-set Multi-cover}, and \emph{Hypergraph Multi-matching} problems, which are natural generalizations of the standard \emph{Stable Set}, \emph{Vertex Cover}, \emph{Dominating Set}, \emph{Set Cover}, and \emph{Maximum Matching} problems.
}


\maketitle              

\tableofcontents

\section{Introduction}\label{intro_sec}

Let a polyhedron $\PC$ 
be defined by one of the following ways: 
\begin{enumerate}
\item[(i)] {\bf System in the canonical form:} 
\begin{equation}\tag{Canon-Form}\label{canonical_form}
    \PC = \{x \in \RR^n \colon A x \leq b\},\quad \text{where $A \in \ZZ^{m \times n}$ and $b \in \QQ^{m}$;}
\end{equation}
\item[(ii)] {\bf System in the standard form:} 
\begin{equation}\tag{Standard-Form}\label{standard_form}
    \PC = \{x \in \RR_{\geq 0}^n \colon A x = b\},\quad \text{where $A \in \ZZ^{k \times n}$ and $b \in \QQ^{k}$.}
\end{equation}
\end{enumerate} 
If $\PC$ is defined by a system in the form \ref{standard_form} with an additional constraint $x \leq u$, for given $u \in \ZZ^n_{\geq 0}$, we call such a system as {\bf the system in the standard form with box constraints}. We consider the following problems:
\begin{problem}[Feasibility]
\begin{equation}
\text{Find a point $x$ inside $\PC \cap \ZZ^n$ or declare that $\PC \cap \ZZ^n = \emptyset$.}\tag{Feasibility-IP}\label{FProb}
\end{equation}
\end{problem}

\begin{problem}[Counting]
\begin{equation}
    \text{Compute the value of $\abs{\PC \cap \ZZ^n}$ or declare that $\abs{\PC \cap \ZZ^n} = +\infty$.}\tag{Count-IP}\label{CProb}
\end{equation}
\end{problem}

\begin{problem}[Optimization]
Given $c \in \ZZ^n$, compute some $x^* \in \PC \cap \ZZ^n$, such that
\begin{equation}\tag{Opt-IP}\label{OptProb}
    c^\top x^* = \max\{c^\top x \colon x \in \PC \cap \ZZ^n\}.
\end{equation}
Or declare that $\PC \cap \ZZ^n = \emptyset$ or that the maximization problem is unbounded.
\end{problem}

\begin{problem}[Optimization and Counting]
Given $c \in \ZZ^n$, compute the number of $x^*$, such that 
\begin{equation}\tag{Opt-And-Count-IP}\label{OptCProb}
    c^\top x^* = \max\{c^\top x \colon x \in \PC \cap \ZZ^n\},
\end{equation}
and find an example of $x^*$, if such exists. Or declare that $\PC \cap \ZZ^n = \emptyset$ or that the maximization problem is unbounded.
\end{problem}

In our work, we analyze these problems under the assumption that the matrix $A$ is sparse. To estimate the sparsity of $A$, it is convenient to use the maximum number of non-zero elements in rows and columns of $A$:
$$
\rowSparse(A) := \max_i \|A_{i *}\|_0 \quad\text{and}\quad \colSparse(A) := \max_j \|A_{* j}\|_0.
$$ Here, $\|x\|_0 = \abs{\{i \colon x_i \not= 0\}}$ denotes the number of non-zeros in a vector $x$ and $A_{i *}$, $A_{* j}$ denote the $i$-th row and the $j$-th column of $A$, respectively. Additionally, we define the \emph{total sparsity of $A$} as the minimum of the above parameters:
$$
\Sparse(A) = \min\bigl\{\rowSparse(A),\, \colSparse(A)\bigr\}.
$$
For our purposes, we sometimes need to use slightly weaker parameters that estimate the number of non-zero elements in non-degenerate square sub-matrices. The reason is that the matrix $A$ can have duplicate rows and columns in some problem definitions. We want to avoid these multiplicities, estimating the sparsity of the matrices. For arbitrary $A \in \ZZ^{m \times n}$, we define
\begin{gather*}
\wrowSparse(A) := \max\{\rowSparse(B) \colon \text{$B$ is non-degenerate sub-matrix of $A$}\},\\    
\wcolSparse(A) := \max\{\colSparse(B) \colon \text{$B$ is non-degenerate sub-matrix of $A$}\}\quad\text{and}\\
\wSparse(A) = \min\bigl\{\wrowSparse(A),\, \wcolSparse(A)\bigr\}.
\end{gather*}
Clearly, the new sparsity parameters are more general than the standard $\rowSparse(A)$ and $\colSparse(A)$:
$$
\wSparse(A) \leq \Sparse(A).
$$
Other parameters that are useful in expressing our results and have some connections with sparsity are matrix norms.
We recall the definitions. The maximum absolute value of entries of a matrix $A$ (also known as \emph{the matrix $\max$-norm}) is denoted by $\|A\|_{\max} = \max_{i,j} \abs{A_{i\,j}}$. For a matrix $A$, by $\|A\|_p$ we denote the matrix norm, induced by the $l_p$ vector norm. It is known that
\begin{gather*}
    \|A\|_1 = \max_{i}\|A_{i *}\|_1 = \max_i \sum_j \abs{A_{i j}}  \quad\text{and}\\
    \|A\|_{\infty} = \max_{j} \|A_{* j}\|_1 = \max_j \sum_i \abs{A_{i j}}.\\
\end{gather*}
Again, we need a similar definition of a norm with respect to non-degenerate sub-matrices $B$ of $A$:
$$
    \gamma_p(A) = \max\{\|B\|_p \colon \text{$B$ is a non-degenerate sub-matrix of $A$}\}.
$$
Surprisingly, the maximum number of vertices in polyhedra defined by systems in the canonical or the standard forms with a fixed $A$ and varying $b$ is also closely connected with sparsity parameters of the matrix $A$ (see Lemma \ref{vertex_num_lm}). The corresponding matrix parameter is denoted by $\nu(A)$:
\begin{gather*}
\nu(A) = \max\limits_{b \in \QQ^m}\abs{\vertex\bigl(\PC(A,b)\bigr)},\quad\text{where}\\
\PC(A,b) = \{x \in \RR^n \colon A x \leq b\}\quad\text{or}\quad\PC(A,b) = \{x \in \RR_{\geq 0}^n \colon A x = b\}.
\end{gather*}
The last important matrix parameters that will be used in our paper are the values of matrix sub-determinants. These parameters are related to sparsity via the Hadamard's inequality.
\begin{definition}
For a matrix $A \in \ZZ^{m \times n}$, by $$
\Delta_k(A) = \max\left\{\abs{\det (A_{\IC \JC})} \colon \IC \subseteq \intint m,\; \JC \subseteq \intint n,\; \abs{\IC} = \abs{\JC} = k\right\},
$$ we denote the maximum absolute value of determinants of all the $k \times k$ sub-matrices of $A$. Here, the symbol $A_{\IC \JC}$ denotes the sub-matrix of $A$, which is generated by all the rows with indices in $\IC$ and all the columns with indices in $\JC$. Note that $\Delta_1(A) = \|A\|_{\max}$. The maximum absolute value of sub-determinants of all orders is denoted by $\Delta_{tot}(A)$, i.e. $\Delta_{tot}(A) = \max_k \Delta_k(A)$. By $\Delta_{\gcd}(A,k)$, we denote the greatest common divisor of determinants of all the $k \times k$ sub-matrices of $A$. Additionally, let $\Delta(A) = \Delta_{\rank(A)}(A)$ and $\Delta_{\gcd}(A) = \Delta_{\gcd}(A,\rank(A))$. The matrix $A$ with $\Delta(A) \leq \Delta$, for some $\Delta > 0$, is called \emph{$\Delta$-modular}.
\end{definition}
Due to the Hadamard's inequality and since $\det(B) = \det(B^\top)$ and $\|x\|_2 \leq \|x\|_1$, for any $B \in \RR^{n \times n}$ and $x \in \RR^n$, the following inequalities connect $\Delta_k(A)$, the matrix norms, and $\Sparse(A)$:
\begin{gather}
    \Delta_k(A) \leq \min\bigl\{\gamma_1(A), \gamma_{\infty}(A)\}^k \leq \min\bigl\{\|A\|_1, \|A\|_{\infty}\}^k,\label{hadamrd_norm_eq0}\\
    \Delta_k(A) \leq \bigl(\|A\|_{\max}\bigr)^k \cdot \wSparse(A)^{k/2} \leq \bigl(\|A\|_{\max}\bigr)^k \cdot \Sparse(A)^{k/2} \label{hadamard_sparse_eq}.
\end{gather}
Denoting $\totn(A) = \min\bigl\{\gamma_1(A), \gamma_{\infty}(A)\}$, the inequality \eqref{hadamrd_norm_eq0} becomes
\begin{equation}\label{hadamrd_norm_eq}
    \Delta_k(A) \leq \totn(A)^k.
\end{equation}

For the reader's convenience, we have put all the notations used in our work into the separate Table \ref{legend_tb}.
Additionally, for the sake of simplicity, in the remainder of the paper we will use the following short notations with respect to the definitions \ref{canonical_form} and \ref{standard_form}: $\Delta := \Delta(A)$, $\Delta_i := \Delta_i(A)$, for $i \in \intint n$, $\Delta_{tot} := \Delta_{tot}(A)$, $\nu := \nu(A)$, $\Delta_{\gcd} := \Delta_{\gcd}(A)$, $\gamma_p := \gamma_p(A)$, for $p \in \intint{+\infty}$, $\totn := \totn(A)$, $\wrowSparse := \wrowSparse(A)$, $\wcolSparse := \wcolSparse(A)$, $\wSparse := \wSparse(A)$.

\begin{table}[h!]
    \caption{Global and Specific Notations}
    \label{legend_tb}
    
    \begin{tabular}[t]{||>{$}l<{$}|l||}
    \hline
    Notation: & Description: \\
    \hline
    \hline
    \rowSparse(A) & Maximum number of non-zeroes in rows of $A$: \\
    & $\rowSparse(A) = \max_i \|A_{i *}\|_0$. \\
    
    \hline

    \colSparse(A) & Maximum number of non-zeroes in columns of $A$: \\
    & $\rowSparse(A) = \max_j \|A_{* j}\|_0$. \\

    \hline

    \Sparse(A) & The \emph{total sparsity of $A$} defined as \\
    & $\Sparse(A) = \min\bigl\{\rowSparse(A),\colSparse(A)\bigr\}.$ \\

    \hline

    \wrowSparse(A) & Maximum number of non-zeroes in rows of non-degenerate sub-matrices of $A$: \\
    & $\wrowSparse(A) = \max\bigl\{\rowSparse(B) \colon \text{$B$ is a non-degenerate sub-matrix of $A$}\bigr\}$. \\
    
    \hline

    \wcolSparse(A) & Maximum number of non-zeroes in columns of non-degenerate sub-matrices of $A$: \\
    & $\wcolSparse(A) = \max\bigl\{\colSparse(B) \colon \text{$B$ is a non-degenerate sub-matrix of $A$}\bigr\}$. \\

    \hline

    \wSparse(A) & The \emph{total sparsity of $A$} with respect to non-degenerate sub-matrices of $A$, \\
    & defined as $\wSparse(A) = \min\bigl\{\wrowSparse(A),\wcolSparse(A)\bigr\}.$ \\

    \hline

    \gamma_p(A) & The maximum $\|\cdot\|_p$-norm of non-degenerate sub-matrices of $A$: \\
    & $\gamma_p(A) = \max\{\|B\|_p \colon \text{$B$ is a non-degenerate sub-matrix of $A$}\}.$ \\

    \hline 

    \totn(A) & $\totn(A) = \min\bigl\{\gamma_1(A), \gamma_{\infty}(A)\}$. \\

    \hline

    \nu(A) & The maximum number of vertices in polyhedra \\
    & with a fixed matrix $A$ and a varying r.h.s. $b$: \\
    & $\nu(A) = \max\limits_{b \in \QQ^m}\abs{\vertex\bigl(\PC(A,b)\bigr)}$, where \\
& $\PC(A,b) = \{x \in \RR^n \colon A x \leq b\}\quad\text{or}\quad\PC(A,b) = \{x \in \RR_{\geq 0}^n \colon A x = b\}.$ \\

    \hline

    \Delta_k(A) & The maximum absolute value of $k \times k$ sub-determinants of $A$: \\
    & $\Delta_k(A) = \max\left\{\abs{\det (A_{\IC \JC})} \colon \IC \subseteq \intint m,\; \JC \subseteq \intint n,\; \abs{\IC} = \abs{\JC} = k\right\}$. \\

    \hline

    \Delta(A) & The maximum absolute value of rank-order sub-determinants of $A$: \\
    & $\Delta(A) = \Delta_{\rank(A)}(A)$. \\

    \hline

    \Delta_{tot}(A) & The maximum absolute value of all sub-determinants of $A$: \\
    & $\Delta_{tot}(A) = \max_k \Delta_k(A)$. \\

    \hline

    \Delta_{\gcd}(A) & The greatest common divisor of rank-order sub-determinants of $A$. \\

    \hline

    \phi & The input bit-encoding length of a corresponding computational problem. \\

    \hline

    \disc(A) & The discrepancy of $A$:\\
    & $\disc(A) = \min\limits_{z \in \{-1/2,\, 1/2\}^n} \left\| A z  \right\|_\infty$. \\

    \hline

    \herdisc(A) & The hereditary discrepancy of $A$:\\
    & $\herdisc(A) = \max\limits_{\IC \subset \intint n} \disc(A_{\IC})$. \\

    \hline

    \detlb(A) & $\detlb(A) = \max_{t} \sqrt[t]{\Delta_t(A)}$. \\

    \hline 

    \numv & The number of vertices in a corresponding hypergraph,\\ 
    & i.e. $\numv = \abs{\VC}$, for a hypergraph $\HC = (\VC, \ES)$. \\

    \hline

    \nume & The number of hyperedges in a corresponding hypergraph,\\ 
    & i.e. $\nume = \abs{\ES}$, for a hypergraph $\HC = (\VC, \ES)$. \\

    \hline 

    \degv & The maximum vertex degree of a corresponding hypergraph, \\
    & $\degv = \max_{v \in \VC} \deg(v)$, for a hypergraph $\HC = (\VC, \ES)$, \\
    & where $\deg(v)$ counts only unique (non-parallel) edges that are incident to $v$. \\

    \hline

    \dege & The maximum hyperedge cardinality of a corresponding hypergraph, \\
    & i.e. $\dege = \max_{\EC \in \ES} \abs{\EC}$, for a hypergraph $\HC = (\VC, \ES)$. \\

    \hline 
    \Delta, \Delta_i, \Delta_{tot}, \nu, \Delta_{\gcd}, & The short notations with respect to the corresponding matrix $A$:\\
    \gamma_p, \totn, \wrowSparse, \wcolSparse, \wSparse & $\Delta := \Delta(A)$, $\Delta_i := \Delta_i(A)$, $\Delta_{tot} := \Delta_{tot}(A)$,\\
    & $\nu := \nu(A)$, $\Delta_{\gcd} := \Delta_{\gcd}(A)$, $\gamma_p := \gamma_p(A)$, $\totn := \totn(A)$, \\
    & $\wrowSparse := \wrowSparse(A)$, $\wcolSparse := \wcolSparse(A)$, $\wSparse := \wSparse(A)$.\\

    \hline
    \hline
    \end{tabular}
\end{table}

\section{Results on Sparse ILP Problems and the Related Work}\label{main_result_sec}

\subsection{General ILP Problems}

Very recently a major breakthrough has been occurred in the ILP complexity theory: based on the works \cite{DadushDis,DadushFDim,VoronoiSlicing} due to Dadush, Peikert \& Vempala and \cite{reverse_Minkowski} due to Regev \& Stephens-Davidowitz, V.~Reis and T.~Rothvoss have proven in \cite{log_ILP} that the problem \ref{OptProb} can be solved in $\log(n)^{O(n)} \cdot \poly(\phi)\footnote{The notation $\phi = \size(A,b,c)$ denotes the input bit-encoding length.}$-time beating the previous $O(n)^n \cdot \poly(\phi)$-time state of the art algorithm due to Dadush, Peikert \& Vempala \cite{DadushDis,DadushFDim}. Note that the complexity results of the works \cite{log_ILP,DadushDis,DadushFDim} are valid for even more general IP problems, where one needs to optimize a convex function defined by the subgradient oracle on a convex region defined by the strict hyperplane separation oracle. Surprisingly, due to Basu \& Oertel \cite{Centerpoints}, the ILP complexity in the oracle-model is $2^{O(n)} \cdot \poly(\phi)$. There exist some more general formulations of IP problems that allow polynomial algorithms in fixed dimension, see for example \cite{Convic,ConvicComp,DConvic}. It is a long-standing open problem to provide a $2^{O(n)} \cdot \poly(\phi)$-time ILP algorithm.

The asymptotically fastest algorithm for the problem \ref{CProb} in a fixed dimension can be obtained, using the approach of A.~Barvinok \cite{Barv_Original} with modifications, due to Dyer \& Kannan \cite{OnBarvinoksAlg_Dyer} and Barvinok \& Pommersheim \cite{BarvPom}. A complete exposition of Barvinok's approach can be found in \cite{BarvBook,BarvPom,BarvWoods,continuous_discretely,AlgebracILP}, additional discussions and connections with "dual"-type counting algorithms could be found in the book \cite{counting_Lasserre_book}, due to J.~Lasserre. An important notion of the \emph{half-open sign decomposition} and other variants of Barvinok's algorithm that is more efficient in practice is given by K\"oppe \& Verdoolaege in \cite{HalfOpen}. The paper \cite{BarvWoods} of Barvinok \& Woods gives important generalizations of the original techniques and adapts them to a wider range of problems to handle projections of polytopes.
Using the fastest deterministic Shortest Lattice Vector Problem (SVP) solver by Micciancio \& Voulgaris \cite{SVP_exp}, the computational complexity of Barvinok's algorithm can be upper bounded by
\begin{equation}\label{BarvComplexityNu}
\nu \cdot 2^{O(n)} \cdot \bigl( \log_2(\Delta) \bigr)^{n \log(n)}.
\end{equation}
Here we can parameterize by $\nu$, because any polyhedron can be transformed to an integer-equivalent simple polyhedron, using a slight perturbation of the r.h.s. vector $b$ (see, for example, Theorem \ref{poly_simplification_th}, due to Megiddo \& Chandrasekaran \cite{epsilon_perturb} and Remark \ref{poly_simplification_rm}). Let us assume that $\PC$ is defined by the form \ref{canonical_form}. Due to the seminal work of P.~McMullen~\cite{MaxFacesTh}, the number of vertices attains its maximum at the class of polytopes, which is dual to the class of cyclic polytopes.  Together with the formula from \cite[Section~4.7]{Grunbaum} for the number of facets of a cyclic polytope, it follows that the maximum number of vertices in an $n$-dimensional polyhedron with $m$ facets is bounded by $$\xi(n,m) = \begin{cases}
    \frac{m}{m-s} \binom{m-s}{s},\text{ for }n = 2s\\
    2\binom{m-s-1}{s},\text{ for }n = 2s+1\\
    \end{cases} = O\left(\frac{m}{n}\right)^{n/2}.$$ Therefore, $\nu \leq \xi(n,m)$ and $\nu = O(m/n)^{n/2}$. Due to \cite[Chapter~3.2, Theorem~3.2]{Schrijver}, we have $\Delta = 2^{O(\phi)}$. Using the notation $\phi$, the bound \eqref{BarvComplexityNu} becomes
\begin{equation}\label{BarvComplexity}
O\Bigl(\frac{m}{n}\Bigr)^{n/2} \cdot \bigl( \log_2(\Delta) \bigr)^{n \log (n)} \quad = \quad O\Bigl(\frac{m}{n}\Bigr)^{n/2} \cdot \phi^{n \log (n)},
\end{equation} which gives a polynomial-time algorithm in a fixed dimension for the problem \ref{CProb}.

The papers \cite{MultyKnapsack_Grib,Counting_FPT_Delta,OnCanonicalProblems_Grib,CountingFixedM} deal with the parameter $\Delta$ to give pseudo-polynomial algorithms, which will be more effective in a varying dimension. Recently, it was shown by Gribanov and Malyshev in \cite{Counting_FPT_Delta} that the \ref{CProb} problem can be solved with an algorithm whose computational complexity is polynomial in $\nu$, $n$, and $\Delta$. Unfortunately, the paper \cite{Counting_FPT_Delta} contains an inaccuracy, which makes its main conclusion incorrect. This inaccuracy was eliminated in \cite{Counting_FPT_Delta_corrected}. The main result of \cite{Counting_FPT_Delta_corrected} (and \cite{Counting_FPT_Delta}) is represented by the following
\begin{theorem}[Gribanov, Shumilov \& Malyshev\cite{Counting_FPT_Delta_corrected}]\label{FasterCounting0_th}
Let $\PC$ be a polytope, given by a system in the standard or the canonical forms and $d := \dim(\PC)$. Then, the problem \ref{CProb} can be solved by a randomized algorithm with the expected arithmetic complexity bound
$$
O\bigr(\nu^2 \cdot d^4 \cdot \Delta^4 \cdot \log_2(\Delta)\bigl).
$$
\end{theorem}
We improve the last result in Theorem \ref{FasterCounting_th} of our paper, and it will be our main tool for sparse problems. A fully self-contained proof of this theorem will be given in Subsection \ref{main_th_0_proof}. Wherever it will be necessary to refer to the original article with an inaccuracy, we will cite the full proof of the relevant statement in Appendix.
\begin{theorem}\label{FasterCounting_th}
Under assumptions of Theorem \ref{FasterCounting0_th}, the problem \ref{CProb} can be solved by a randomized algorithm with the expected arithmetic complexity bound:
$$
O\bigl(\nu^2 \cdot d^4 \cdot \Delta^3\bigr).
$$
\end{theorem}
Using Theorem \ref{FasterCounting0_th} and different ways to estimate $\nu$, the paper \cite{Counting_FPT_Delta} gives new interesting arithmetic complexity bounds for the \ref{FProb} and \ref{CProb} problems. Let us present them, taking into account the improvement made in the Theorem \ref{FasterCounting_th}:
\begin{itemize}

\item The bound $$
    O\Bigl(\frac{m}{n}\Bigr)^{n} \cdot n^4 \cdot \Delta^3$$ for systems in the form \ref{canonical_form} that is polynomial in $m$ and $\Delta$, for any fixed $n$. In comparison with the bound \eqref{BarvComplexity}, this bound has a much better dependence on $n$, considering $\Delta$ as a parameter. For example, taking $m = O(n)$ and $\Delta = 2^{O(n)}$, the above bound becomes $2^{O(n)}$, which is even faster, than the state of the art algorithm for the problem \ref{FProb}, due to Reis \& Rothvoss \cite{log_ILP}, with the complexity bound $\log(n)^{O(n)}\cdot \poly(\phi)$;
    
    \item The general bound, for systems in the canonical or the standard forms,
    $$
    O(n)^{4 + n} \cdot \Delta^{3+2n}
    $$ that is polynomial on $\Delta$, for any fixed $n$;
    
    \item The bound 
    \begin{equation}\label{FixedM_counting_complexity}
    O\bigl(n/k\bigr)^{2k} \cdot n^4 \cdot \Delta^3
    \end{equation}
    for systems in the form \ref{standard_form}, which is also valid for systems in the form \ref{canonical_form} with $k = m - n$,
    that is polynomial on $n$ and $\Delta$, for $k = O(1)$. Taking $k = 1$, it gives an $O\bigl( n^6 \cdot \Delta^3 \bigr)$-algorithm to compute the number of integer points in a simplex. The last result can be used to count solutions of the Unbounded Subset-Sum problem, which is formulated as follows. Given numbers $w_1, \dots, w_n$ and $W$, we need to count the number of ways to exchange the value $W$ by the values $w_i$, assuming that each value $w_i$ can be used unlimitedly. It can be done by algorithms with the arithmetic complexity bound 
    $$
    O( n^6 \cdot w_{\max}^3 ).
    $$ Moreover, this result can be used to handle the $k$-dimensional variant of the Unbounded Subset-Sum problem, when the costs $w_i$ and $W$ are represented by $k$-dimensional vectors. Using the Hadamard's bound, it gives the following arithmetic complexity bound:
    $$
    O(n)^{2(k+2)} \cdot k^{-k/2} \cdot w_{\max}^{3 k},
    $$ where $w_{\max} = \max_{i} \|w_i\|_{\infty}$. Note that the earlier paper of Lasserre \& Zeron \cite{knapsack_lasserre} also gives a counting FPT-algorithm for the Unbounded Subset-Sum problem, parameterized by $w_{\max}$, but an exact complexity bound was not given. 
\end{itemize}

In the current work, we try to estimate the value of $\nu$ in a different way, to handle ILP problems with sparse matrices. Additionally, we generalize Theorem \ref{FasterCounting_th} to work with the problem \ref{OptCProb}. The resulting theorem is the following:
\begin{theorem}\label{main_th_1}
Let $\PC$ be a polyhedron, defined by the system in \ref{canonical_form}. Then, the problems \ref{FProb} and \ref{CProb} can be solved by an algorithm, whose complexity can be estimated by the following formulas
\begin{gather*}
\bigl(\totn\bigr)^{5 n} \cdot 4^n \cdot \poly(\phi),\\
\bigl(\|A\|_{\max}\bigr)^{5n} \cdot \bigl(\wSparse\bigr)^{3.5 n} \cdot 4^n \cdot \poly(\phi).
\end{gather*}
The problem \ref{OptCProb} can be solved by an algorithm, whose complexity can be estimated by the following formulas (under the assumption that $c \not= \BZero$)
\begin{gather*}
    \bigl(\totn\bigr)^{7 n} \cdot \bigl(\|c\|_{\infty}\bigr)^3 \cdot 2^{4n}  \cdot \poly(\phi),\\
    \bigl(\|A\|_{\max}\bigr)^{7n} \cdot \bigl(\|c\|_{\infty}\bigr)^3 \cdot \bigl(\wSparse\bigr)^{5.5 n} \cdot 2^{4n} \cdot \poly(\phi).
\end{gather*}
\end{theorem} 
The theorem's proof is given in Subsection \ref{main_th_1_proof}. This new complexity bounds, applied to the problems in the form \ref{canonical_form}, are emphasized in Table \ref{sparse_results_tb}. 
As the reader could see, with respect to the problem \ref{FProb}, under the assumptions $\gamma_{1,\infty} \leq \log^{\varepsilon}(n)$ or $\|A\|_{\max} \leq \log^{\varepsilon}(n)$ and $\wSparse \leq \log^{\varepsilon}(n)$, for some $\varepsilon > 0$, our complexity bounds outperform the state of the art complexity bound $\log(n)^{O(n)} \cdot \poly(\phi)$. With respect to the problem \ref{CProb}, under the assumption $\|A\|_{\max} = n^{o(\log(n))}$,
our complexity bounds outperform the state of the art complexity bound $O(m/n)^{n/2} \cdot \phi^{n \log(n)}$.
\begin{table}[h!]
    \caption{The complexity bounds for the problems \ref{FProb}, \ref{CProb}, \ref{OptProb}, and \ref{OptCProb} in the form \ref{canonical_form}}
    \label{sparse_results_tb}
    
    \begin{tabular}{||m{9em}|c|m{10em}||}
    \hline
    \hline
    Problems: & Time:\footnotemark[1] & Reference: \\
    \hline
    \hline
    \ref{FProb} and \ref{OptProb} & $\log(n)^{O(n)}$ & Reis \& Rothvoss \cite{log_ILP} \\
    \hline
    \ref{CProb} & $O\bigl(m/n\bigr)^{n/2} \cdot \phi^{n \log(n)}$ & Barvinok et al. \cite{BarvPom,Barv_Original,OnBarvinoksAlg_Dyer}\\
    \hline
    \ref{FProb} and \ref{CProb} & $\bigl(\totn\bigr)^{5n} \cdot 4^n$ & {\color{red} this work} \\
    & $\bigl(\|A\|_{\max}\bigr)^{5n} \cdot \wSparse(A)^{3.5n} \cdot 4^n$ &\\
    \hline
    \ref{OptCProb} & $\bigl(\totn\bigr)^{7n} \cdot \bigl(\|c\|_{\infty}\bigr)^3 \cdot 2^{4n}$ & {\color{red} this work} \\
    & $\bigl(\|A\|_{\max}\bigr)^{7n} \cdot \bigl(\|c\|_{\infty}\bigr)^3 \cdot \bigl(\wSparse\bigr)^{5.5n} \cdot 2^{4n}$ &\\
    \hline
    \hline
    \end{tabular}

\footnotetext[1]{The multiplicative factor $\poly(\phi)$ is skipped.}

\end{table}
The following corollary, which is a straightforward consequence of Theorem \ref{main_th_1}, shows that, under some assumptions, the \ref{CProb} and \ref{OptCProb} problems can be solved by a faster algorithm than the complexity bound \eqref{BarvComplexity} gives.
\begin{corollary}\label{main_corr_1}
In the notation of Theorem \ref{main_th_1}, assuming that $\|A\|_{\max} = n^{O(1)}$ and $\|c\|_{\infty} = n^{O(n)}$, the problems \ref{CProb} and \ref{OptCProb} can be solved by algorithms with the complexity bound
$
n^{O(n)} \cdot \poly(\phi).
$
\end{corollary}

\subsubsection{About our Method}

The current paper continues the series of works \cite{Counting_FPT_Delta_corrected,Counting_FPT_Delta,CountingFixedM}, which are aimed to present efficient pseudopolynomial algorithms for the problems \ref{CProb} and \ref{OptCProb}, based on using rational generating functions together with the seminal Brion's theorem. As it was already mentioned, this approach was used by Barvinok in his seminal work \cite{Barv_Original} to present the first polynomial-time in a fixed dimension algorithm for the problems \ref{CProb} and \ref{OptCProb}.

The most important feature of a new approach, introduced in \cite{Counting_FPT_Delta_corrected,Counting_FPT_Delta}, is that we do not compute the rational generating function of the set $\PC \cap \ZZ^n$. Instead of doing this, we directly compute a compact generating function of the exponential series $\fG(\PC; \tau) = \sum\limits_{z \in \PC \cap \ZZ^n} e^{\langle c , z \rangle}$ that depends only on a single variable $\tau$. The exponential generating function can be obtained from the original rational generating function, substituting $x_i = e^{c_i \tau}$, for some $c \in \RR^n$. The new function forgets the structure of the set $\PC \cap \ZZ^d$, but it is still useful for counting. For example, two monomials $x_1^1 x_2^2$ and $x_1^2 x_2^1$ glue to one exponential term $2 e^{3 \tau}$ after the map $x_i = e^{c_i \tau}$ with $c = (1, 1)^\top$. Our method to compute $\fG(\PC; \tau)$ is based on the Brion's theorem and a novel dynamic programming technique that processes tangent cones of $\PC$. The dynamic programming table is indexed by the dimensionality of the subproblems and the elements of the Gomory group associated with a corresponding tangent cone.

Let us discuss a secondary part of a new method that may also have an independent interest. For a given set $\AC$ of $m$ non-zero vectors in $\QQ^n$, let us consider the problem to compute a vector $z \in \ZZ^n$, such that $a^\top z \not= 0$, for all $a \in \AC$. Preferably, the value of $\|z\|_{\infty}$ should be as small as possible. Due to the original work of A.~Barvinok \cite{Barv_Original}, the vector $z$ could be found by a polynomial-time algorithm as a point on the moment curve. The paper \cite{HalfOpen} of K{\"o}ppe \& Verdoolaege gives an alternative method, based on "irrational decompositions" from the work \cite{IrrationalDecomp} of K{\"o}ppe. These polynomial-time methods can generate $z$ with the only guaranty $\|z\|_{\infty} \leq M^n$, for some constant $M \geq m$. However, due to De~Loera, Hemmecke, Tauzer \& Yoshida \cite{EffectiveCounting}, the vector $z$ with sufficiently small components can be effectively chosen by a randomized algorithm. Unfortunately, the paper \cite{EffectiveCounting} does not give exact theoretical bounds that are needed to develop pseudopolynomial algorithms. In turn, the paper \cite{Counting_FPT_Delta_corrected} presents a new and very simple randomized polynomial-time algorithm that generates the desired vector $z$ with $\|z\|_{\infty} \leq \abs{\AC}$. The precise description of this fact is emphasized in Theorem \ref{all_non_zero_th}.

Compared to the previous papers \cite{Counting_FPT_Delta_corrected,Counting_FPT_Delta} in the series, the current paper gives a more efficient dynamic programming computational scheme. Additionally, we give a new bound on the number of vertices of a rational polyhedron that is helpful to prove our complexity bounds and can have an independent interest.

\subsubsection{Other Related Work on Sparse and $\Delta$-modular ILPs}

Due to Kratsch \cite{PolyKernelSparse}, the sparse ILP problems attain a polynomial kernalization with respect to the parameter $n + u$, where $u$ is the maximum variable range. More precisely, it was shown that any ILP can be reduced to an equivalent ILP with $O(u^r \cdot n^r)$ variables and constraints with the coefficients bit-encoding length $O(\log (n u))$, where $r := \rowSparse(A)$. On the contrary, if the range $u$ is unbounded, then $r$-row-sparse ILP problems do not admit a polynomial kernelization unless $NP \subseteq coNP/poly$.

There are many other interesting works about the ILP's complexity with respect to the parameter $\Delta$. Since a good survey is given in the work \cite{OnCanonicalProblems_Grib}, we mention only the most remarkable results. The first paper that discovers fundamental properties of the bimodular ILP problem ($\Delta = 2$) is \cite{BimodularVert}, due to Veselov \& Chirkov. Using results of \cite{BimodularVert}, a strong polynomial-time solvability of the bimodular ILP problem was proved by Artmann, Weismantel \& Zenklusen in \cite{BimodularStrong}. 
Unfortunately, not much is known for $\Delta \geq 3$. Very recently, it was shown by Fiorini, Joret, Weltge \& Yuditsky in \cite{TwoNonZerosStrong} that the ILP problem is polynomial-time solvable, for any fixed $\Delta$, if the matrix $A$ has at most $2$ non-zeros per row or per column. Previously, a weaker result, based on the same reduction, was known, due to Alekseev \& Zakharova \cite{AZ}. It states that any ILP with a $\{0,1\}$-matrix $A$, which has at most two non-zeros per row and a fixed value of $\Delta\binom{\BUnit^\top}{A}$, can be solved by a linear-time algorithm.

Additionally,  we  note  that,  due  to Bock, Faenza, Moldenhauer \& Ruiz-Vargas \cite{StableSetHardness}, there  are  no  polynomial-time  algorithms  for  the ILP problems with $\Delta = \Omega(n^\varepsilon)$, for any $\varepsilon > 0$, unless  the ETH (the Exponential Time Hypothesis) is false. 
The last fact is the reason why we need to use both parameters $\nu$ and $\Delta$. Due to \cite{StableSetHardness}, the complexity bound $\poly(\Delta, \phi)$ is unlikely to exist, while the bound $\poly(\nu, \Delta, \phi)$ is presented in Theorem \ref{FasterCounting_th}, which is used in Theorem \ref{main_th_1} to develop efficient algorithms for sparse problems.

\subsection{ILP Problems with a Bounded Co-dimension}\label{standard_ILP_subs}

In this subsection, we consider ILP problems in the form \ref{standard_form}. Since in our definition $k = \rank(A)$, it is essential to call the parameter $k$ as the \emph{co-dimension of $\PC$.}  We are interested in the complexity bounds for bounded values of $k$. Let us survey some remarkable results. The following result, due to Gribanov et al. \cite[see Theorem~8 and Corollary~9]{OnCanonicalProblems_Grib}, gives a parameterization by $k$ and $\Delta$.
\begin{theorem}[Gribanov et al. \cite{OnCanonicalProblems_Grib}]\label{FasterILPFixedM_th}
Assume that some $k \times k$ non-degenerate sub-matrix $B$ of $A$ is given and $\eta = \Delta/\abs{\det(B)}$. Then, the problem \ref{OptProb} can be solved by an algorithm with the arithmetic complexity bound
$$
O(k)^{k+1} \cdot \eta^{2k} \cdot \Delta^2 \cdot \log(\Delta_{\gcd}) \cdot \log(k \cdot \Delta).
$$
\end{theorem}
As it was noted in \cite{OnCanonicalProblems_Grib}, due to \cite{SubdeterminantApprox}, we can assume that $\eta = O(\log (k))^k$, and the previous complexity bound becomes
$$
O\bigl(\log(k)\bigr)^{2 k^2} \cdot k^{k+1} \cdot \Delta^2 \cdot \log(\Delta_{\gcd}) \cdot \log(k \cdot \Delta).
$$

For the case when $A$ only has non-negative elements, the basic dynamic-programming scheme from \cite{Bellman} can be used to derive an algorithm, parameterized by $\|b\|_{\infty}$ and $k$. Using fast $(\min,+)$-convolution algorithms (see, for example, \cite{StructuredMinPlus} or \cite{KnapsackSubsetSum_SmallItems}), the same complexity bound can be used for systems in the \ref{standard_form} form with box constraints. We emphasize it in the following statement:
\begin{proposition}\label{b_DP_th}
The problem \ref{OptProb} in the form \ref{standard_form} with box constraints can be solved by an algorithm with the arithmetic complexity bound
$$
O\Bigl(n \cdot \bigl(\|b\|_{\infty}+1\bigr)^k \Bigr).
$$
\end{proposition}
Due to the works \cite{BrunchWidthILP} and \cite{FPT_ILP_Optimality_Fomin} of Cunningham \& Geelen and Fomin et al., the parameter $k$ in the term $\bigl(\|b\|_{\infty}+1\bigr)^k$ can be replaced by stronger parameters $2 \omega$ or $\rho + 1$, where $\omega$ is the branch-width and $\rho$ is the path-width of the column matroid of $A$. 

The approach, which is most important for us in this Subsection, is based on the notion of the \emph{hereditary discrepancy of $A$}. 
\begin{definition}
For a matrix $A \in \RR^{k \times n}$, \emph{its discrepancy and its hereditary discrepancy} are defined by the formulas
\begin{gather*}
\disc(A) = \min_{z \in \{-1/2,\, 1/2\}^n} \left\| A z  \right\|_\infty,\\
\herdisc(A) = \max_{\IC \subset \intint n} \disc(A_{* \IC}).
\end{gather*}
\end{definition}
The paper \cite{OnIPAndConv}, due to Jansen and Rohwedder, gives a powerful ILP algorithm, parameterized by $\herdisc(A)$ and $k$, which will be our second main tool.
\begin{theorem}[Jansen \& Rohwedder \cite{OnIPAndConv}]\label{DiscrILP_th}
Let $H = \herdisc(A)$ and assume that there exists an optimal solution $x^*$ of the problem \ref{OptProb} with $\|x^*\|_1 \leq K$. Then, the problem \ref{OptProb} can be solved by an algorithm with the complexity bound
$$
O(H)^{2k} \cdot \log(K).
$$
\end{theorem}
Different bounds on $\herdisc(A)$ can be used to develop different complexity bounds for ILP problems. Due to the works \cite{HerDisc} and \cite{SixDeviations_Spencer} of Lov\'asz,  Spencer, \& Vesztergombi, and Spencer, it is known that
\begin{equation}\label{SixDeviations_eq}
\herdisc(A) \leq 2 \disc(A) \leq \eta_k \cdot \|A\|_{\max}, \quad\text{where}\quad \eta_k \leq 12 \cdot \sqrt{k}.
\end{equation}
Due to Beck and Fiala \cite{DiscBeckBound}, the value of $\herdisc(A)$ is bounded by the $l_1$-norm of columns. More precisely,
\begin{equation}\label{Beck_eq}
    \herdisc(A) < \|A\|_{\infty}.
\end{equation}
Additionally, Beck and Fiala conjectured that $\herdisc(A) = O\bigl(\sqrt{ \|A\|_{\infty} }\bigr)$ and settling this has been an elusive open problem. The best known result in this direction is due to Banaszczyk \cite{BanaszDiscBound}:
\begin{equation}\label{Banasz_eq}
    \herdisc(A) = O\Bigl( \sqrt{\|A\|_{\infty} \cdot \log(n)} \Bigr).
\end{equation}
The important matrix characteristic that is closely related to $\herdisc(A)$ is $\detlb(A)$. Due to Lov\'asz, Spencer, \& Vesztergombi \cite{HerDisc}, it can be defined as follows:
$$
\detlb(A) = \max\limits_{t \in \intint k} \sqrt[t]{\Delta_t(A)},
$$ and it was shown in \cite{HerDisc} that 
$
\herdisc(A) \geq (1/2) \cdot \detlb(A)
$. Matou\v{s}ek in \cite{DiscDetBound} showed that $\detlb(A)$ can be used to produce tight upper bounds on $\herdisc(A)$. The result of Matou\v{s}ek was improved by Jiang \& Reis in \cite{TightDiscDetBound}:
\begin{equation}\label{DiscDetBound_eq}
\herdisc(A) = O\Bigl( \detlb(A) \cdot \sqrt{\log(k) \cdot \log(n)} \Bigr).    
\end{equation} 

Next, let us consider the problems \ref{CProb} and \ref{OptCProb}. Clearly, the number of vertices in a polyhedron, defined by a system in the \ref{canonical_form}, can be estimated by $\binom{n}{k} = O(n/k)^k$. The last fact in combination with Theorem \ref{FasterCounting_th} results in the following corollary, which gives a parameterization by $\Delta$ and $k$.
\begin{corollary}\label{fast_counting_fixed_k_cor}
Assume that $\PC$ is bounded, then the problem \ref{CProb} can be solved by an algorithm with the arithmetic complexity bound
$
O(n/k)^{2k} \cdot (n-k)^4 \cdot \Delta^3
$.
\end{corollary}

\begin{remark}\label{extended_fixed_k_complexity_rm}
Note that if we already know an optimal solution $x^*$ of the problem \ref{OptProb}, we can solve the problem \ref{OptCProb}, using Corollary \ref{fast_counting_fixed_k_cor} just by adding the equality $c^\top x = c^\top x^*$ to the problem's definition. Clearly, the resulting arithmetic complexity bound is 
\begin{equation}\label{extended_fixed_k_complexity_eq}
    O(n/k)^{2(k+1)} \cdot (n-k)^4 \cdot \bigl(\|c\|_{\infty}\bigr)^3 \cdot \Delta^3.
\end{equation}
\end{remark}

The next theorem considers the ILP problems in the standard form with sparse $A$. In this theorem, we just summarize the combinations of Theorem \ref{DiscrILP_th} with the different bounds on $\herdisc(A)$. Additionally, we use Corollary \ref{fast_counting_fixed_k_cor} to solve the counting-type problems. Note that the $5$-th complexity bound of the next theorem has already been proven in \cite{OnIPAndConv}, we put it here for the sake of completeness.
\begin{theorem}\label{main_th_2}
Let $\PC$ be a polyhedron, defined by the form \ref{standard_form}. The problems \ref{FProb} and \ref{OptProb} can be solved by algorithms with the following complexity bounds:
\begin{enumerate}
    \item $O\bigl(\gamma_{\infty}\bigr)^{2k} = O\bigl(\|A\|_{\max}\bigr)^{2k} \cdot \bigl(\wcolSparse\bigr)^{2k}$,
    
    \item $O\bigl(\gamma_{\infty}\bigr)^{k} \cdot 2^{k \cdot \log \log(n)} = O\bigl(\|A\|_{\max}\bigr)^{k} \cdot \bigl(\wcolSparse\bigr)^{k} \cdot 2^{k \cdot \log \log(n)}$,
    
    \item $O\bigl(\gamma_{1}\bigr)^{2k} \cdot 2^{k \cdot \log (\log(k) \cdot \log(n))}$,
    
    \item $O\bigl(\|A\|_{\max}\bigr)^{2k} \cdot \bigl(\wrowSparse\bigr)^{k} \cdot 2^{k \cdot \log \bigl(\log(k) \cdot \log(n)\bigr)}$,
    
    \item $O\bigl(\|A\|_{\max}\bigr)^{2k} \cdot k^{k}$.
\end{enumerate}
The problem \ref{CProb} can be solved by algorithms with the following complexity bounds:
\begin{enumerate}
    \setcounter{enumi}{5}
    \item $O(n/k)^{2 k} \cdot \bigl(\totn\bigr)^{3 k}$,
    \item $O(n/k)^{2 k} \cdot \bigl(\|A\|_{\max}\bigr)^{3 k} \cdot \bigl(\wSparse\bigr)^{1.5 k}$.
\end{enumerate}
The problem \ref{OptCProb} can be solved by the same algorithm with the cost of an additional multiplicative term $\bigl(\|c\|_{\infty}\bigr)^3$ in the complexity bound. Everywhere in the complexity bounds, we skip the $\poly(\phi)$ multiplicative term. 
\end{theorem}
\begin{proof}
Due to Theorem \ref{DiscrILP_th}, the problems \ref{FProb} and \ref{OptProb} can be solved by algorithms with the arithmetic complexity bound $O(H)^{2k} \cdot \log(K)$, where $H = \herdisc(A)$ and $K = \|x^*\|_1$, for any optimal solution $x^*$. It is known that the problem has an optimal solution $x^*$ with $size(x^*) = \poly(\phi)$, so $\log(K) = \poly(\phi)$.

Now, the $1$-st bound follows from the inequality \eqref{Beck_eq}. The $2$-nd bound follows from the inequality \eqref{Banasz_eq}. To establish the $3$-rd and the $4$-th bounds, we use the equality \eqref{DiscDetBound_eq}. Due to the inequalities \eqref{hadamard_sparse_eq} and \eqref{hadamrd_norm_eq}, we clearly have
$
    \detlb(A) \leq \|A\|_{\max} \cdot \sqrt{\wrowSparse}, \quad\text{and}\quad \detlb(A) \leq \|A\|_1
$. Putting these bounds to \eqref{DiscDetBound_eq}, it gives the $3$-rd and the $4$-th complexity bounds. The $5$-th complexity bound directly follows from the inequality \eqref{SixDeviations_eq}.

Now, let us consider the problems \ref{CProb} and \ref{OptCProb}. The $6$-th and $7$-th complexity bounds straightforwardly follow from the bounds \eqref{hadamrd_norm_eq}, \eqref{hadamard_sparse_eq} respectively and Corollary \ref{fast_counting_fixed_k_cor}. To satisfy its prerequisites, $\PC$ needs to be bounded. If $\PC$ is unbounded, then we can check that $\abs{\PC \cap \ZZ^n} = 0$, using the algorithm for the problem \ref{FProb}. As it was already mentioned, its complexity can be estimated by $O\bigl(\|A\|_{\max}\bigr)^{2k} \cdot k^{k}$, which has no effect on the desired bound. In the opposite case, we have $\abs{\PC \cap \ZZ^n} = +\infty$. So, we can assume that $\PC$ is bounded, and the result is true. Note additionally that, if $\PC \cap \ZZ^n \not= \emptyset$, then we can use the same algorithm for the problem \ref{FProb} to find some $x \in \PC \cap \ZZ^n$. Finally, using the same reasoning, the complexity bounds for the problem \ref{OptCProb} just follows from Corollary \ref{fast_counting_fixed_k_cor} and its Remark \ref{extended_fixed_k_complexity_rm}.
\end{proof}

\subsection{ILP problems in the Form \ref{standard_form} with Box-constraints}

Finally, before we will finish the current section, let us consider ILP problems in the form \ref{standard_form} with box constraints. Using the basic dynamic programming scheme from \cite{SteinitzILP}, combined with a linear-time algorithm for the $(\min,+)$-convolution (see, for example, \cite[Theorem~7]{OnCanonicalProblems_Grib}, \cite{StructuredMinPlus} or \cite{KnapsackSubsetSum_SmallItems}), it is easy to prove the following proposition.
\begin{proposition}\label{fixed_k_ILP_compl_chi_th}
    The problem \ref{OptProb} in the form \ref{standard_form} with box constraints can be solved by an algorithm with the arithmetic complexity bound
    $$
O(\chi + k)^{k} \cdot \bigl( \|A\|_{\max}\bigr)^k,
$$ where $\chi$ is a value of the $l_1$-proximity bound. That is
$$
\chi = \max\limits_{x^*} \min\limits_{z^*} \|x^* - z^*\|_1,
$$ where $x^*$ and $z^*$ are optimal solutions of the LP relaxation and of the original ILP, respectively.
\end{proposition}
Different bounds on $\chi$ give different algorithms, based on Proposition \ref{fixed_k_ILP_compl_chi_th}. The paper \cite{SteinitzILP} of Eisenbrand \& Weismantel gives $
\chi \leq k \cdot \bigl(2k \cdot \|A\|_{\max} + 1\bigr)^k
$.
The paper \cite{ProximityViaSparsity}, due to Lee, Paat et al., gives
$
\chi \leq (2k + 1)^k \cdot \Delta
$.
Recent result of Lee, Paat et al. \cite{ModularDiffColumns} states that
$$
 \chi \leq k \cdot (k+1)^2\cdot \Delta^3 + (k+1) \cdot \Delta = O(k^3 \cdot \Delta^3).    
$$
The dependence on $\Delta$ in the last bound can be reduced by Averkov \& Schymura \cite{DiffColumnsOther}
\begin{equation}\label{chi_diffcol_delta_ineq}
\chi = O(k^5 \cdot \Delta^2).
\end{equation}
Using Proposition \ref{fixed_k_ILP_compl_chi_th} with the bound \eqref{chi_diffcol_delta_ineq}, we see that the ILP in the form \ref{standard_form} with box constraints can be solved by an algorithm with the arithmetic complexity bound
\begin{equation*}
    \bigl( \|A\|_{\max}\bigr)^k \cdot O(\Delta)^{2k} \cdot k^{5 k}.
\end{equation*}
Using the inequalities \eqref{hadamrd_norm_eq} and \eqref{hadamard_sparse_eq}, the last bound transforms to the bounds
\begin{gather}
    O(k)^{5k} \cdot (\totn)^{2k^2} = O(\totn)^{2k^2 + O(k \log(k))},\notag\\
    O(k)^{5k} \cdot \bigl(\|A\|_{\max}\bigr)^{2k^2+k} \cdot \bigl(\wSparse\bigr)^{k^2} = \bigl(\|A\|_{\max}\bigr)^{2k^2+k} \cdot \bigl(\wSparse\bigr)^{k^2 + O(k \log(k))}.\label{fixed_k_ILP_compl_sparse_th}
\end{gather}
In Table \ref{fixed_k_ILP_tb}, we summarize all the facts, mentioned in the current subsection. The complexity bounds for the problems \ref{FProb}, \ref{OptProb}, \ref{CProb}, \ref{OptCProb} without box constraints are taken from Theorem \ref{main_th_2} and Remark \ref{extended_fixed_k_complexity_rm}. To handle the problems with box constraints, we just take the complexity bound \eqref{fixed_k_ILP_compl_sparse_th}. We also mention that the existence of algorithms for the problems \ref{CProb} and \ref{OptCProb} in the form \ref{standard_form} with box constraints, parameterized by $k$ and polynomial by $n$, is open, and it is a good direction for further research.

\begin{table}[h!]
    \caption{New complexity bounds for ILP problems in the form \ref{standard_form}}
    \label{fixed_k_ILP_tb}
    
    \begin{tabular}[t]{||c|>{$\qquad}p{25em}<{\qquad$}||}
    \hline
    \hline
    Problems: & Time:\footnotemark[1] \\
    \hline
    \hline
    \ref{OptProb} without mult. & O(\gamma_{\infty})^{2k} = O\bigl(\|A\|_{\max}\bigr)^{2k} \cdot \bigl(\wcolSparse\bigr)^{2k} \\
    
    & O(\gamma_{\infty})^k \cdot 2^{k \cdot \log\log(n)} = O\bigl(\|A\|_{\max}\bigr)^{k} \cdot \bigl(\wcolSparse\bigr)^k \cdot 2^{k \cdot \log\log(n)} \\
    
    &O(\gamma_1)^{2k} \cdot 2^{k \cdot \log\bigl(\log(k) \cdot \log(n)\bigr)} \\
    
    & O\bigl(\|A\|_{\max}\bigr)^{2k} \cdot \bigl(\wrowSparse\bigr)^k \cdot 2^{k \cdot \log\bigl(\log(k) \cdot \log(n)\bigr)} \\
    
    & O\bigl(\|A\|_{\max}\bigr)^{k} \cdot k^k, \quad \text{\color{red} due to Jansen \& Rohwedder \cite{OnIPAndConv}} \\
    \hline
    
    \ref{CProb} without mult.\,\footnotemark[2]$\quad$ & O(n/k)^{2 k} \cdot \bigl(\totn\bigr)^{3k} \\
    
    & O(n/k)^{2 k} \cdot O\bigl(\|A\|_{\max}\bigr)^{3k} \cdot \bigl(\wSparse\bigr)^{1.5 k} \\
    
    \hline 
    
    \ref{OptProb} with mult. & O\bigl(\totn\bigr)^{2k^2 + O(k \log k)} \\
    & O\bigl(\|A\|_{\max}\bigr)^{2k^2 + k} \cdot \bigl(\wSparse\bigr)^{k^2 + O(k \log k)} \\
    
    \hline 
    
    \ref{CProb} with mult. & \text{\color{red} open problem}\\
    
    \hline
    \hline
    \end{tabular}

    \footnotetext[1]{The multiplicative factor $\poly(\phi)$ is skipped.}
    \footnotetext[2]{To solve the problem \ref{OptCProb}, we need to pay an additional multiplicative factor $\bigl(\|c\|_{\infty}\bigr)^3$.}
\end{table}

    

\section{Applications: The Vertex/Edge Multi-Packing and Multi-Cover Problems on Graphs and Hypergraphs.}\label{comb_prob_sec}

To define a hypergraph, we will often use the notation $\HC = (\VC, \ES)$, where $\VC$ is the set of vertices, represented by an arbitrary finite set, and $\ES \subseteq 2^{\VC}$ is a set of hyperedges. To denote a single vertex and a single hyperedge of $\HC$, we will use the symbols $v \in \VC$ and $\EC \in \ES$. Additionally, we denote $\numv = \abs{\VC}$, $\nume = \abs{\ES}$, $\degv = \max_{v \in \VC} \deg(v)$, and $\dege = \max_{\EC \in \ES} \abs{\EC}$. In other words, the symbols $\numv$, $\nume$, $\degv$, and $\dege$ denote the number of vertices, the number of hyperedges, the maximum vertex degree, and the maximum edge cardinality, respectively. We use this notation to avoid ambiguity with the notation $n$, $m$, and $d$ from the subsections, considering ILP problems.

In some problem formulations, we need to deal with hypergraphs $\HC = (\VC, \ES)$ having parallel hyperedges. That is, $\ES$ is a multi-set of sets $\EC \in 2^{\VC}$. In this case, by $\deg(v)$ we denote the number of {\bf unique} hyperedges that are incident to $v$, and $\degv$ denotes the maximum vertex degree with respect to {\bf unique} hyperedges.

In our work, we consider two types of combinatorial multi-packing/multi-cover problems: \emph{vertex-based} problems and \emph{edge-based} problems. In vertex-based problems, we need to pack vertices into hyperedges or to cover the hyperedges by vertices. In edge-based problems, we need to pack hyperedges or to cover vertices by hyperedges. The word "multi" means that we can choose a multi-set of vertices or edges to satisfy cover constraints or to not violate packing constraints. Before we give formal definitions, we present a few examples. The \emph{Stable Multi-set} problem, which was introduced by Koster and Zymolka in \cite{stable_multiset} as a natural generalization of the standard \emph{Stable Set} problem, is an example of a vertex-based multi-packing problem. Similarly, the \emph{Vertex Multi-cover} problem, which is a natural generalization of the standard \emph{Vertex Cover} problem, is an example of a vertex-based multi-cover problem. Some properties of the Stable Multi-set problem polyhedron were investigated in \cite{stable_multiset_poly,stable_multiset_cycles}, which had given a way to construct effective branch \& bound algorithms for this problem. We cannot find a reference to the paper that introduces the Vertex Multi-cover problem, but this problem can be interpreted as a blocking problem for the Stable Multi-set problem (introduction to the theory of blocking and anti-blocking can be found in \cite{blocking_and_antib,anti_blocking}, see also \cite[p.~225]{geometric_algorithms}).

The examples of edge-based problems are the \emph{Set Multi-cover}, \emph{Multi-set Multi-cover}, and \emph{Hypergraph Multi-matching} problems. The Set Multi-cover problem is a natural generalization of the classic \emph{Set Cover} problem, where we need to choose a multi-set of hyperedges to cover the vertices by a given number of times. In the Multi-Set Multi-Cover problem, the input hypergraph $\HC$ can have parallel hyperedges. This problem has received quite a lot of attention in the recent papers \cite{elections,ILP_application_multicovering_voting,approx_multicover,exact_multicover_paper,exact_multicover_proceed,combinatorial_nfold}. An exact $O\bigl((c_{\max}+1)^{\numv} \cdot \nume\bigr)$ arithmetic complexity algorithm for the Multi-Set Multi-Cover problem, parameterized by $\numv$ and the maximum coverage constraint number $c_{\max}$, is given by Hua, Wang, Yu \& Lau in \cite{exact_multicover_paper,exact_multicover_proceed}. A double exponential $2^{2^{O(\numv \log \numv)}} \cdot \poly(\phi)$-complexity FPT-algorithm, parameterized by $n$, is given in Bredereck et al. \cite{elections}. The last algorithm was improved to a $\numv^{O(\numv^2)} \cdot \poly(\phi)$-complexity algorithm by Knop, Kouteck\`y \& Mnich in \cite{combinatorial_nfold}. A polynomial-time approximation algorithm can be found in Gorgi et al. \cite{approx_multicover}. The \emph{Hypergraph Multi-matching} problem is a very natural generalization of the \emph{Hypergraph Matching} problem (see, for example, \cite{OnHyperMatchingAlon,GeometricMatchingBook}), which in turn is a generalization of the standard \emph{Maximum Matching} problem in simple graphs. We cannot find a reference to the paper that formally introduces the Hypergraph Multi-matching problem, but, again, this problem can be interpreted as a blocking problem for the Multi-set Multi-cover problem. The papers \cite{ILP_application_multicovering_voting,FiveMiniatures} give good surveys and contain new ideas to use the ILP theory in combinatorial optimization setting.

Now, let us give some formal definitions. The vertex-based multi-packing/multi-cover problems can easily be modeled, using the following template problem:
\begin{problem}[Hypergraph Vertex-Based Multi-packing/Multi-cover]\label{VertexHyper_pr}
Let $\HC = (\VC, \ES)$ be a hypergraph. Given numbers $c_{\EC}$, $p_{\EC} \in \ZZ_{\geq 0}$, for $\EC \in \ES$, compute a multi-subset of $\VC$, represented by natural numbers $x_{v}$, for $v \in \VC$, such that
\begin{enumerate}
    \item[(i)] $c_{\EC} \leq x(\EC) \leq p_{\EC}$, for any $\EC \in \ES$;
    \item[(ii)] $x(\VC)$ is maximized or minimized.
\end{enumerate}
Here, $x(\MC) = \sum_{v \in \MC} x_v$, for any $\MC \subseteq \VC$. In other words, we need to solve the following ILP:
\begin{align}
    &\max\bigl\{\BUnit^\top x\bigr\} \;\text{ or }\; \min\bigl\{\BUnit^\top x\bigr\}\notag\\
    &\begin{cases}
    c \leq A(\HC)^\top x \leq p\\
    x \in \ZZ_{\geq 0}^{\VC},
    \end{cases}\tag{Vertex-Based-ILP}\label{VertexHyper_ILP}
\end{align}
where $A(\HC)$ denotes the vertex-hyperedge incidence matrix of $\HC$, and the vectors $c$ and $p$ are composed of the values $p_{\EC}$ and $c_{\EC}$, respectively. 
It is natural to think that $\HC$ does not contain parallel hyperedges, because the multiple edge-constraints can easily be replaced by a stronger one.

If $c_{\EC} = -\infty$, for all $\EC \in \ES$, and $x(\VC)$ is maximized, it can be considered as the \emph{Stable Multi-set Problem on Hypergraphs}, when we need to find a multi-set of vertices of the maximum size, such that each hyperedge $\EC \in \ES$ is triggered at most $p_{\EC}$ times.
Similarly, if $p_{\EC} = +\infty$, for all $\EC \in \ES$, and $x(\VC)$ is minimized, it can be considered as the \emph{Vertex Multi-cover Problem on Hypergraphs}, when we need to find a multi-set of vertices of the minimum size, such that each hyperedge $\EC \in \ES$ is triggered at least $c_{\EC}$ times. 

For the case, when $\HC$ is a simple graph, these problems can be considered as very natural generalizations of the classical Stable Set and Vertex Cover problems. Following \cite{stable_multiset}, the first one is called the \emph{Stable Multi-set Problem}. As it was previously discussed, it is natural to call the second problem as the \emph{Vertex Multi-cover Problem}.
\end{problem}

\begin{definition}\label{vertex_problems_def_rm}
Given numbers $u_{v} \in \ZZ_{\geq 0}$, for $v \in \VC$, we can add additional constraints $x_{v} \leq u_v$ to any of the problems above. We call such a problem as a \emph{problem with multiplicities}.
Similarly, given $w_v \in \ZZ$, for $v \in \VC$, we can consider the objective function $\sum_{v \in \VC} w_v x_v$ instead of $x(\VC) = \sum_{v \in \VC} x_v$. We call such a problem as a \emph{weighted problem}. The maximum weight is denoted by $w_{\max} = \max_{v \in \VC} \abs{w_v}$.
\end{definition}
Similarly, the edge-based multi-packing/multi-cover problems can be modeled using the following template problem:
\begin{problem}[Hypergraph Edge-Based Multi-packing/Multi-cover]\label{EdgeHyper_pr}
Let $\HC = (\VC, \ES)$ be a hypergraph. Given numbers $c_{v},p_v \in \ZZ_{\geq 0}$, for $v \in \VC$, compute a multi-subset of $\ES$, represented by the natural numbers $x_{\EC}$, for $\EC \in \ES$, such that
\begin{enumerate}
    \item[(i)] $c_v \leq x\bigl(\delta(v)\bigr) \leq p_{v}$, for any $v \in \VC$;
    \item[(ii)] $x(\ES)$ is maximized or minimized.
\end{enumerate}
Here, $x(\MS) = \sum_{\EC \in \MS} x_{\EC}$, for any $\MS \subseteq \ES$, and $\delta(v) = \{\EC \in \ES \colon v \in \EC\}$ denotes the set of hyperedges that are incident to the vertex $v$.

The problem can be represented by the following ILP:
\begin{align}
    &\max\bigl\{\BUnit^\top x\bigr\} \;\text{ or }\; \min\bigl\{\BUnit^\top x\bigr\} \notag\\
    &\begin{cases}
    c \leq A(\HC) x \leq p\\
    x \in \ZZ_{\geq 0}^{\ES},
    \end{cases}\tag{Edge-Based-ILP}\label{EdgeHyper_ILP}
\end{align}
where the vectors $c$, $p$ are composed of the values $c_{v}$ and $p_{v}$.
Again, it is natural to think that $\HC$ does not contain parallel hyperedges, because the multiple edge-variables can be easily glued to one variable.

If $c_{v} = -\infty$, for all $v \in \VC$, and $x(\ES)$ is maximized, it can be considered as the \emph{Hypergraph Multi-matching} problem, when we need to find a multi-set of hyperedges of the maximum size, such that each vertex $v \in \VC$ is triggered at most $p_{v}$ times.
Similarly, if $p_{v} = +\infty$, for all $v \in \VC$, and $x(\ES)$ is minimized, it can be considered as the \emph{Set Multi-cover} problem, when we need to find a multi-set of hyper-edges of the minimum size, such that each vertex $v \in \VC$ is triggered at least $c_{v}$ times. 

For the case, when $\HC$ is a simple graph, these problems can be considered as very natural generalizations of the classical Matching and Edge Cover problems. It seems natural to call these problems as the \emph{Maximum Multi-matching} and \emph{Edge Multi-cover} problems. The definition of the \emph{Edge Multi-cover} problem can be found, for example, in the work \cite{EdgeMulticoverApprox}, due to Cohen and Nutov. For the \emph{Maximum Multi-matching} problem, we did not find a correct reference.

Similarly, we can introduce the \emph{Dominating Multi-set Problem} on simple graphs, which is a natural generalization of the classical \emph{Dominating Set} problem. In this problem, we need to find a multi-set of vertices of the minimal size, such that all the vertices of a given graph will be covered given number of times by neighbors of the constructed vertex multi-set. The Dominating Multi-set Problem can be straightforwardly reduced to the Set Multi-cover Problem. To do that, we just need to construct the set system $\HC = (\VC, \ES)$, where $\VC$ coincides with the set of vertices of a given graph, and $\ES$ is constituted by neighbors of its vertices.
\end{problem}

\begin{definition}\label{edge_problems_def_rm}
By analogy with Definition \ref{vertex_problems_def_rm}, we introduce the \emph{weighted} variants and variants \emph{with multiplicities} for the all edge-based multi-packing/multi-cover problems discussed above. Note that 
the presence of parallel edges for these problems is not redundant and makes the corresponding problem more general.
The weighted Set Multi-cover with multiplicities is known in literature as the \emph{Weighted Multi-set Multi-cover} problem, see, for example,  \cite{exact_multicover_paper,exact_multicover_proceed,combinatorial_nfold}.
\end{definition}

Let us explain our motivation with respect to the specified combinatorial problems. The classical Stable Set and Vertex Cover Problems on graphs and hypergraphs admit trivial $2^{O(\numv)} \cdot \poly(\phi)$-complexity algorithms. However, the Stable Multi-set and Vertex Multi-cover Problems do not admit such a trivial algorithm. But, both problems can be modeled as the ILP problem \eqref{VertexHyper_ILP} with $\numv$ variables. Consequently, both problems can be solved by the previously mentioned $\log(\numv)^{O(\numv)} \cdot \poly(\phi)$-complexity general ILP algorithm. Here $\phi = \size(c,p,w,u)$.

Is it possible to give a faster algorithm? Is it possible to give a positive answer to this question, considering a more complex variant with multiplicities? We show that these problems on hypergraphs can be solved by a $\min\{\degv,\dege\}^{O(\numv)} \cdot \poly(\phi)$-complexity algorithm. Consequently, the Stable Multi-set and Vertex Multi-cover Problems on simple graphs can be solved by $2^{O(\numv)} \cdot \poly(\phi)$-complexity algorithms. Our complexity results for these problems, together with the Multi-set Multi-cover, Hypergraph Multi-matching, and Dominating Multi-set problems, are gathered in Theorem \ref{main_th_3}.

\begin{theorem}\label{main_th_3}
Let us consider the \ref{OptCProb}-variants of the problems Stable Multi-set, Vertex Multi-cover, Set Multi-cover, Hypergraph Multi-matching, and Dominating Multi-set with multiplicities (also known as the Multi-set Multi-cover problem). The following complexity bounds hold:    
    \begin{tabular}[t]{||p{21em}|>{$\qquad}c<{\qquad$}||}
    \hline
    \hline
    Problems: & Time: \\
    \hline
    \hline
    
    Stable Multi-set and Vertex Multi-cover on hypergraps & \min\{\degv,\dege\}^{5.5 \numv} \cdot 2^{4 \numv} \\
    
    \hline
    
    Stable Multi-set and Vertex Multi-cover on simple graphs & 2^{9 \numv} \\
    
    \hline
    
    Dominating Multi-set & {\degv}^{5.5 \numv} \cdot 2^{4 \numv} \\
    
    \hline
    
    Set Multi-cover and Hypergraph Multi-matching & \min\{\degv,\dege\}^{5.5 \nume} \cdot 2^{4 \nume} \\
    
    \hline
    \hline
    \end{tabular}
The complexity bounds for the weighted variants of the considered problems contain an additional multiplicative term $w_{\max}^3$. Everywhere in the complexity bounds, we skip the $\poly(\phi)$ multiplicative term. 
\end{theorem}
\begin{proof}
To prove the theorem, we use Theorem \ref{main_th_1} for the problems' definitions: \ref{VertexHyper_pr}, \ref{EdgeHyper_pr}, \ref{vertex_problems_def_rm} and \ref{edge_problems_def_rm}. This approach gives us the desired complexity bounds for all the problems, except for the Stable Multiset and Vertex Multicover problems on simple graphs. For these exceptions, we will give a more refined analysis. 

We follow the proof of Theorem \ref{main_th_1}, using a more refined bound for $\Delta_{\numv-1}$ and $\Delta_{\numv}$, where $A := A(\GC)$ be the incidence matrix of the corresponding simple graph $\GC$. Due to Grossman, Kulkarni \& Schochetman \cite{minors_incidence_matrix}, the absolute values of sub-determinants of a simple graph incidence matrix can be bounded in terms of \emph{the odd tulgeity of $\GC$}. More precisely, $\Delta_i \leq 2^{\tau_0}$, where $\tau_0 = \tau_0(\GC)$ is \emph{the odd tulgeity of $\GC$}, which is defined as the maximum number of vertex-disjoint odd cycles of $\GC$. Clearly, $\tau_0 \leq \numv/3$, so, 
$$
\max\bigl\{\Delta_{\numv-1},\Delta_{\numv} \bigr\} \leq 2^{\numv/3}.
$$ 
Using this bound in the proof of Theorem \ref{main_th_1}, it gives the desired complexity bounds for the Stable Multiset and Vertex Multicover Problems.
\end{proof}

\subsection{The Multi-set Multi-cover and Hypergraph Multi-matching Problems Parameterized by the Number of Vertices $\numv$.}\label{comb_n_param_subs}

In Theorem \ref{main_th_3}, we have presented $\min\{\degv,\dege\}^{5.5 \nume} \cdot 2^{4 \nume} \cdot \poly(\phi)$-complexity algorithms for the \ref{OptCProb}-variant of the Set Multi-cover and Hypergraph Multi-matching problems with multiplicities. Due to Knop, Kouteck\`y \& Mnich \cite{combinatorial_nfold}, the weighted \ref{OptProb}-variants of these problems admit an ${\numv}^{O(\numv^2)} \cdot \poly(\phi)$-complexity algorithm, which is faster than our algorithm for $\nume = \Omega({\numv}^{2+\varepsilon})$ and any $\varepsilon > 0$. In other words, our last complexity bound is good only for sufficiently sparse hypergraphs.

So, the motivation of this subsection is to present faster algorithms for the \ref{OptCProb}- and \ref{OptProb}-variants of the weighted Multi-set Multi-cover and Hypergraph Multi-matching problems with and without multiplicities, parameterized by $\numv$ instead of $\nume$.  Our results for these problems are gathered in Table \ref{multiset_results_tb}.

\begin{table}[h!]

    \caption{New complexity bounds for the Set Multi-cover and Hypergraph Multi-matching problems}
    \label{multiset_results_tb}
    
    \begin{tabular}[t]{||>{\qquad}c<{\qquad}|>{$\qquad}c<{\qquad$}||}
    \hline
    \hline
    Version: & Time:\footnotemark[1] \\
    \hline
    \hline
    
    \ref{OptProb}, without multiplicities & O(\dege)^{2 \numv} \\
    
     & O(\dege)^{\numv} \cdot 2^{\numv \cdot \log(\dege \log (\numv))} \\
     
     & O(\degv)^n \cdot 2^{\numv \cdot \log(\log(\degv \numv)\log(\numv))}\\
     
     & O(\numv)^{\numv}\\
     
    \hline
    
    \ref{OptCProb}, without multiplicities \qquad & \min\{\dege,\degv\}^{1.5 \numv} \cdot O(\nume/\numv)^{2 \numv} \cdot  w_{\max}^3 \\
    
    &{\dege}^{1.5 \numv} \cdot O(\numv)^{2 \dege\numv + O(\dege)} \cdot w_{\max}^3 \\
    
    &{O(\degv)}^{3.5 \numv} \cdot w_{\max}^3 \\
    
    & 4^{\numv^2 + O(\numv)} \cdot w_{\max}^3 \\

    \hline

    \ref{OptProb}, with multiplicities & O(\min\{\degv,\dege\})^{\numv^2 + O(\numv \log \numv)} \\

    \hline

    \ref{OptCProb}, with multiplicities & \text{\color{red} open problem} \\

    \hline
    \hline
    \end{tabular}

    \footnotetext[1]{The multiplicative factor $\poly(\phi)$ is skipped.}

\end{table}

\begin{remark}\label{ParamNComplexity_rm} Let us have a little discussion about the complexity bounds, presented in Table \ref{multiset_results_tb}. Firstly, let us consider the problems without multiplicities. As the reader can see, for fixed $\dege$, the weighted \ref{OptProb}-variant of the considered problems can be solved by $2^{O(\numv)}$-complexity algorithms (the $\poly(\phi)$-term is ignored). For $\dege = \log(\numv)^{O(1)}$, the best complexity bound is $2^{O(\numv \cdot \log \log (\numv))}$. Another interesting case is $\degv = o(\numv)$, which gives the $o(\numv)^{\numv} \cdot 2^{O(\numv \cdot \log \log (\numv))}$-complexity bound. For other values of parameters, the general $O(\numv)^{\numv}$-complexity bound holds. 

For the unweighted \ref{OptCProb}-variant of the considered problems, if $\dege$ is fixed, then the ${\numv}^{O(\numv)}$-complexity algorithm exists. The same is true if $\degv = \numv^{O(1)}$ or $\nume = \numv^{O(1)}$. The complexity $2^{O(n)}$ is possible, if a hypergraph has constant maximum degree $\degv = O(1)$ or, if it is very sparse $\nume = O(\numv)$ and has a constant maximum hyperedge cardinality $\dege = O(1)$. For the general values of $\dege$, $\degv$, and $\nume$, it is better to use the complexity bound $\min\{\degv,\dege\}^{1.5 \numv} \cdot O(\nume/\numv)^{2 \numv}$. Since $\nume \leq 2^{\numv}$, it straightforwardly gives the general $4^{\numv^2 + O(\numv)}$-complexity bound. Note that the considered complexity bounds for the problems without multiplicities sufficiently outperform the best complexity bound that we know $n^{O(n^2)}$, due to Knop, Kouteck\`y, and Mnich \cite{combinatorial_nfold}.

Now, let us consider the problems with multiplicities. Note again that the weighted Set Multi-cover problem with multiplicities is also known as the weighted Multi-set Multi-cover problem. In comparison with the state of the art complexity bound $\numv^{O(\numv^2)}$, our bound $O\bigl(\min\{\degv,\dege\}\bigr)^{\numv^2 + O(\numv \log \numv)}$ has a lower exponent base, and it gives a constant-estimate in the exponent power. Unfortunately, we are not able to present a complexity bound, parameterized by $\numv$, for the \ref{OptCProb}-variant, and it seems to be an interesting open problem.
\end{remark}

We omit proofs of the results, presented in Table \ref{multiset_results_tb}, because they straightforwardly follow from the complexity bounds, described in Theorem \ref{main_th_2} and Table \ref{fixed_k_ILP_tb}. Indeed, the weighted Multi-set Multi-cover and Hypergraph Multi-matching problems with or without multiplicities can be represented by the following ILP's in the standard form:
\begin{equation*}
\begin{gathered}
\max\bigl\{w^\top x\bigr\}\\
\begin{cases}
\bigl( A(\HC)\; I_{\numv \times \numv} \bigr) x = p\\
\BZero \leq x \leq u\\
x \in \ZZ^{\numv+\nume}
\end{cases}
\end{gathered} 
\quad    
\begin{gathered}
\min\bigl\{w^\top x\bigr\}\\
\begin{cases}
\bigl( -A(\HC)\; I_{\numv \times \numv} \bigr) x = -c\\
\BZero \leq x \leq u\\
x \in \ZZ^{\numv+\nume},
\end{cases}
\end{gathered}
\end{equation*}
where the constraint $x \leq u$ needs to be omitted for the variants without multiplicities. The co-dimension of these formulations is $\numv$. Using simple bounds $\nume \leq 2^{\numv}$, $\nume \leq \degv \numv$, and $\nume = O(\numv)^{\dege+1}$ that are valid for the problems without multiplicities, the desired complexity bounds of Table \ref{multiset_results_tb} can be easily obtained.  Note that the equality $\nume = O(\numv)^{\dege+1}$ directly follows from the inequality
$
\nume \leq \sum_{i = 1}^{\dege} \binom{\numv}{i}
$.

\section{Additional Notes: Expected ILP Complexity}\label{exp_ILP_complexity_sec}
It was shown by Oertel, Paat \& Weismantel in \cite{IntegralityNumber} that, for almost all r.h.s. $b \in \ZZ^m$, the original ILP problem in the form \ref{canonical_form} is equivalent to the problem $\max\{ c^\top x \colon A_{\BC} x \leq b_{\BC},\, x \in \ZZ^n \}$, where $A_{\BC}$ is a non-degenerate $n \times n$ sub-matrix of $A$, induced by some optimal LP base $\BC$. It was noted by Shevchenko in \cite[Paragraph~3.3., p. 42--43]{BlueBook} (see also \cite[Chapter~5.2]{OnCanonicalProblems_Grib}) that such a square ILP problem is equivalent to the group minimization problem, described by R.~Gomory in the seminal work \cite{GomoryRelation} (see also \cite[Chapter~19]{HuBook}). Consequently, due to \cite{GomoryRelation} and \cite[Chapter~19]{HuBook}, such an ILP can be solved by an algorithm with the arithmetic complexity bound 
\begin{equation}\label{canonical_group_compl_eq}
O\bigl(\min\{n, \Delta\} \cdot \Delta \cdot \log(\Delta)\bigr),
\end{equation}
where $\Delta := \Delta(A)$.
The result of Oertel, Paat \& Weismantel \cite{IntegralityNumber} was refined by Gribanov et al. in \cite[Chapter~5.5.1]{OnCanonicalProblems_Grib}, where a stronger probability argument was given.

A stronger result for the problems in the form \ref{standard_form} is given in the paper of Oertel, Paat \& Weismantel \cite{DistributionsILP}, where the distributions of the corresponding random variables are presented. Another way is to reduce the problem in the form \ref{standard_form} to the problem in  the form \ref{canonical_form}, using \cite[Lemma~5]{OnCanonicalProblems_Grib}. It follows from \cite{GomoryRelation} and \cite{DistributionsILP} (or from \cite[Lemma~5]{OnCanonicalProblems_Grib}) and result for the form \ref{canonical_form}) that, for almost all r.h.s. $b \in \ZZ^k$, the ILP problem in the form \ref{standard_form} of co-dimension $k$ can be solved by an algorithm with the arithmetic complexity bound
\begin{equation}\label{standard_group_compl_eq}
O\bigl( (n-k) \cdot \Delta \cdot \log(\Delta) \bigr).
\end{equation} It is also easy to see that this fact also holds for problems with multiplicities, the simplest way is to reduce the problem into the form \ref{canonical_form}.

The bounds \eqref{canonical_group_compl_eq} and \eqref{standard_group_compl_eq}, together with the inequalities \eqref{hadamrd_norm_eq} and \eqref{hadamard_sparse_eq}, give the following complexity bounds, described in Table \ref{avg_results_gen_tb}, for the sparse problem \ref{OptProb} in the canonical and the standard forms, respectively.
\begin{table}[h!]
    \caption{Expected complexity bounds for almost all $b$ for the problem \ref{OptProb} in the standard and the canonical forms}
    \label{avg_results_gen_tb}
    
    \begin{tabular}[t]{||c|>{$\qquad}c<{\qquad$}||}
    \hline
    \hline
    Problems: & Time:\footnotemark[1] \\
    \hline
    \hline
    The form \ref{canonical_form}, for almost all $b \in \ZZ^m$ & \bigl(\totn\bigr)^n \\
     & \bigl( \|A\|_{\max} \bigr)^{n} \cdot \bigl(\wSparse\bigr)^{n/2}\\
    \hline
    The form \ref{standard_form}, for almost all $b \in \ZZ^k$ & \bigl(\totn\bigr)^k  \\
     &  \bigl( \|A\|_{\max} \bigr)^{k} \cdot \bigl(\wSparse\bigr)^{k/2}\\
    \hline
    \hline
    \end{tabular}
    
    \footnotetext[1]{The multiplicative factor $\poly(\phi)$ is skipped.}
\end{table}
These bounds can be used to give expected-case complexity bounds for the combinatorial problems, described in Table \ref{avg_results_comb_tb}.
\begin{table}[h!]
    \caption{Expected complexity bounds for combinatorial packing/cover problems with multiplicities, for almost all r.h.s. $p$ or $c$}
    \label{avg_results_comb_tb}
    
    \begin{tabular}[t]{||l|>{$\qquad}c<{\qquad$}||}
    \hline
    \hline
    Problems:\footnotemark[1] & Time:\footnotemark[2] \\
    \hline
    \hline
    Stable Multi-set on hypergraps and Hypergraph Multi-matching, & \min\{\degv,\dege\}^{\numv/2}\\
    for almost all r.h.s. $p$ & \\
    \hline
    Vertex Multi-cover on hypergraps and Multi-set Multi-cover, & \min\{\degv,\dege\}^{\numv/2}\\
    for almost all r.h.s. $c$ & \\
    \hline
    Dominating Multi-set, for almost all r.h.s. $c$ & {\degv}^{\numv/2}\\
    \hline
    Stable Multi-set and Vertex Multi-cover on simple graphs, & \\
    for almost all r.h.s. $p,c$ resp. \footnotemark[3] & 2^{\numv/3} \\
    \hline
    \hline
    \end{tabular}
    
\footnotetext[1]{All the considered problems are weighted problems with multiplicities.}
\footnotetext[2]{The multiplicative factor $\poly(\phi)$ is skipped.}
\footnotetext[3]{The bound $2^{\numv/2}$ is trivial, to achieve the bound $2^{\numv/3}$, see the proof of Theorem \ref{main_th_3}.}

\end{table}

\section{Summary of the Paper and Open Problems}\label{conclusion_sec}
Here we give a summary of results, notes, and implications of our work.
\begin{itemize}
    \item[$\bullet$] We show that the problems \ref{CProb} \& \ref{OptCProb} with respect to sparse instances with bounded elements and their weaker versions \ref{FProb} \& \ref{OptProb} can be solved by algorithms that outperform the general state of the art $\log(n)^{O(n)} \cdot \poly(\phi)$-complexity algorithm for \ref{OptProb}, due to Reis \& Rothvoss \cite{log_ILP}. Details can be found in Table \ref{sparse_results_tb} and Theorem \ref{main_th_1}. For example, if the matrix $A$ is an $\{-1,0,1\}$-matrix, and it has constant number of non-zeroes in each row/column, then the corresponding problems \ref{CProb} \& \ref{OptCProb} can be solved in $2^{O(n)} \cdot \poly(\phi)$-time.
    
    \item[$\bullet$] We show that in the assumptions $\|A\|_{\max} = n^{O(1)}$ and $\|c\|_{\infty} = n^{O(n)}$, the problems \ref{CProb} and \ref{OptCProb} can be solved by algorithms with the complexity bound 
    $
    n^{O(n)} \cdot \poly(\phi)
    $, which outperforms the state of the art bound \eqref{BarvComplexity} for the problems \ref{CProb} and \ref{OptCProb}. For details, see Corollary \ref{main_corr_1}.

    \item[$\bullet$] We give an improved arithmetic complexity bound $O(\nu^2 \cdot n^4 \cdot \Delta^3)$ for the problem \ref{CProb} with respect to the older bound $O\bigl(\nu^2 \cdot n^4 \cdot \Delta^4 \cdot \log(\Delta)\bigr)$, see Theorem \ref{FasterCounting_th}.
    
    \item[$\bullet$] We give new algorithms for the \ref{OptCProb}-variant of the Stable Multi-set, Vertex Multi-cover, Set Multi-cover, Multi-matching, and Dominating Multi-set problems with respect to simple graphs and hypergraphs, see the definitions \ref{VertexHyper_pr} and \ref{EdgeHyper_pr}. The weighted variants and the variants with the multiplicities of the above problems are handled, see Definitions \ref{vertex_problems_def_rm} and \ref{edge_problems_def_rm}. Note that the weighted Set Multi-cover problem with multiplicities is also known as the weighted Multi-set Multi-cover problem. Our algorithms outperform the general state of the art ILP algorithms, applied to these problems. Details can be found in Theorem \ref{main_th_3}.

    \item[$\bullet$] We summarize known results and new methods to give new algorithms for the \mbox{\ref{FProb}-,} \ref{CProb}-, \ref{OptProb}-, \ref{OptCProb}-variants of ILP problems in the standard form with and without multiplicities, parameterized by $\|A\|_{\max}$ and the co-dimension of $A x = b$. The new complexity bounds outperform general-case bounds on sparse instances. Details can be found in Subsection \ref{standard_ILP_subs}, Table \ref{fixed_k_ILP_tb}, and Theorem \ref{main_th_2}.

    \item[$\bullet$] Using our notes for sparse problems in the standard form, we give new algorithms for the \ref{OptProb}- and \ref{OptCProb}-variants of the Set Multi-cover and Hypergraph Multi-matching problems with and without multiplicities, parameterized by the number of vertices $\numv$. The weighted variants are handled. Tighter complexity bounds with respect to the parameters $\numv$, $\nume$, $\dege$, and $\degv$ are considered. 
    Unfortunately, we are not able to present a complexity bound, parameterized by $\numv$, for the \ref{OptCProb}-variant with multiplicities, it seems to be an interesting open problem. Our complexity bounds for the considered problems outperform the state of the art $\numv^{O(\numv^2)} \cdot \poly(\phi)$-complexity bound, due to Knop, Kouteck\`y, and Mnich \cite{combinatorial_nfold}. Details can be found in Subsection \ref{comb_n_param_subs} and Table \ref{multiset_results_tb}. Discussion can be found in Remark \ref{ParamNComplexity_rm}.

\end{itemize}

\subsubsection*{Open Problems:}
\begin{itemize}

    \item[$\bullet$] As it was noted before, we are not able to present an algorithm for the problem \ref{CProb} in the form \ref{standard_form} with multiplicities, which will be polynomial on $n$, $\Delta$ or $\|A\|_{\max}$, for any fixed co-dimension $k$. More precisely, given $A \in \ZZ^{k \times n}$, $b \in \QQ^k$, and $u \in \ZZ^n$, let us consider the polyhedron $\PC$, defined by the system $A x = b,\;0 \leq x \leq u$.
    The problem is to develop an algorithm to compute $\abs{\PC \cap \ZZ^n}$, whose complexity will be polynomial on $n$, $\Delta$ or $\|A\|_{\max}$, for any fixed $k$. Despite considerable effort, we are not able to present such an algorithm. The main difficulty is that our methods work well only in the scenarios, when the value of $\abs{\vertex(\PC)}$ is sufficiently small. But, in the current case, the value of $\abs{\vertex(\PC)}$ can be equal to $2^n$. Note that the positive solution for this problem can grant new more efficient algorithms for the Multi-set Multi-cover problem and its weighted variant.

    \item[$\bullet$] Our general complexity bounds (see Theorems \ref{main_th_1} and \ref{main_th_2}) for sparse variants of the problem \ref{CProb} contain a term of the type $\bigl(\|A\|_{\max}\bigr)^{O(n)}$ or of the type $\bigl(\|A\|_{\max}\bigr)^{O(k)}$. Could we develop an algorithm, which will be polynomial on $\|A\|_{\max}$ and more efficient for sparse problems with respect to the general state of the art algorithms? Could we do this for the simpler problem \ref{FProb}? 

    \item[$\bullet$] Our complexity bounds for sparse problems depend mainly on the total number of variables $n$, which can by significantly bigger than an actual dimension $d = \dim(\PC)$ of a polyhedron. The known state of the art algorithms can be easily adapted to work with the parameter $d$ instead of $n$. For example, the state of the art algorithm, due to Reis \& Rothvoss \cite{log_ILP}, gives the $\log(d)^{O(d)} \cdot \poly(\phi)$ complexity bound. Unfortunately, at the current moment, we can not adapt our methods for sparse problems to work with the parameter $d$. The difficulty is concentrated in Lemma \ref{vertex_num_lm}, which estimates the number of vertices of a polyhedron. The proof of such a lemma, based on a parameter $d$, is an interesting open question, which will guaranty the existence of an algorithm for sparse problems, parameterized by $d$ instead of $n$.
\end{itemize}


\section{Proofs of the Main Theorems \ref{FasterCounting_th} and \ref{main_th_1}}

\subsection{The Smith Normal Form}\label{SNF_subs}

Let $A \in \ZZ^{m \times n}$ be an integer matrix of rank $n$. It is a known fact (see, for example, \cite{Schrijver,SNFOptAlg,Zhendong}) that there exist unimodular matrices $P \in \ZZ^{m \times m}$ and $Q \in \ZZ^{n \times n}$, such that $A = P \dbinom{S}{\BZero_{d \times n}} Q$, where $d = m - n$ and $S \in \ZZ_{\geq 0}^{n \times n}$ is a diagonal non-degenerate matrix. Moreover, $\prod_{i = 1}^{k} S_{ii} = \Delta_{\gcd}(A,k)$, and, consequently, $S_{ii} \mid S_{(i+1) (i+1)}$, for $i \in \intint{n-1}$. The matrix $\dbinom{S}{\BZero_{d \times n}}$ is called the \emph {Smith Normal Form} (or, shortly, the SNF) of the matrix $A$. Near-optimal polynomial-time algorithms for constructing the SNF of $A$ are given in the work \cite{SNFOptAlg} due to Storjohann \& Labahn.

\subsection{Algebra of Rational Polyhedra and Generating Functions}

Let $\VC$ be a Euclidean space with the inner product denoted by $\langle \cdot, \cdot \rangle$. Let $\Lambda \subseteq \VC$ be a lattice and $\Lambda^{\circ}$ be its dual.
\begin{definition}
For a polyhedron $\PC \subseteq \VC$, a vector $c \in \VC$ and an abstract variable $\tau$, we denote
\begin{equation*}
    \fG(\PC,c;\tau) = \sum\limits_{z \in \PC \cap \Lambda} e^{ \langle c, z \rangle \tau}.
\end{equation*}
The \emph{polar} of $\PC$ is denoted by
$
\PC^{\circ} = \{y \in \VC \colon \langle y, x \rangle \leq 1, \forall x \in \PC \}
$.
\end{definition}

\begin{definition}
Let $\AC \subseteq \VC$ be a set. The \emph{indicator} $[\AC]$ of $\AC$ is the function $[\AC]\colon \VC \to \RR$ defined by
$$
[\AC](x) = \begin{cases}
1\text{, if }x \in \AC\\
0\text{, if }x \notin \AC.
\end{cases}
$$ 
\end{definition}

\begin{definition}
The polyhedron $\PC \subseteq \VC$ is called \emph{rational}, if it can be defined by a system of finitely many inequalities 
$$
\langle a_i, x \rangle \leq b_i, \quad \text{where $a_i \in \Lambda^{\circ}$ and $b_i \in \ZZ$.}
$$
 The \emph{algebra of rational polyhedra} $\PS(\QQ\VC)$ is the vector space, defined as the span of the indicator functions of all the rational polyhedra $\PC \subseteq \VC$.
\end{definition}

We recall the following restatement of the theorem proved by Lawrence \cite{Lawrence} and independently by Khovanski \& Pukhlikov \cite{Pukhlikov}, declared as Theorem~13.8b in \cite[Section~13]{BarvBook}.
\begin{theorem}[Lawrence \cite{Lawrence}, Khovanski \& Pukhlikov \cite{Pukhlikov}]\label{Pukhlikov_th}
    Let $\dim(\VC) = n$ and $\RS(\VC)$ be the linear space of functions acting from $\VC$ to $\RR$, spanned by finite linear combinations of the following functions
    $$
    c \quad\to\quad \frac{e^{\langle c, v \rangle}}{\bigl(1 - e^{\langle c, u_1 \rangle}\bigr) \cdot \ldots \cdot \bigl(1- e^{\langle c, u_n \rangle}\bigr)},
    $$
    where $v \in \Lambda$  and $u_i \in \Lambda \setminus\{\BZero\}$, for $i \in \intint n$. Then, there exists a linear transformation 
    $$
    \FC \colon \PS(\QQ\VC) \to \RS(\VC),
    $$ such that the following properties hold:
    \begin{enumerate}
    
        \item[(1)] Let $\PC \subseteq \VC$ be a non-empty rational polyhedron without lines and let $\RC := \RC_{\PC} \subseteq \VC$ be its recession cone. Then, for all $c \in \inter(\RC^{\circ})$, the series 
        $$
        \sum\limits_{z \in \PC \cap \Lambda} e^{\langle c, z \rangle}
        $$
        converges absolutely to a function $\FC\bigl([\PC]\bigr)$.
        
        \item[(2)] If $\PC$ contains a line, then $\FC\bigl([\PC]\bigr) = 0$.
        
    \end{enumerate}
\end{theorem}
Note that hereafter we will use this Theorem \ref{Pukhlikov_th} just with $\VC = \RR^n$ and $\Lambda = \ZZ^n$. The following lemma represents a core of Theorem \ref{FasterCounting_th} and contains a main improvement with respect to the counting algorithm from \cite{Counting_FPT_Delta}.
\begin{lemma}\label{quad_system_lem}
Let $A \in \ZZ^{n \times n}$, $b \in \ZZ^n$, $\Delta = \abs{\det(A)} > 0$. Let us consider the polyhedron $\PC = \{ x \in \RR^n \colon A x \leq b\}$. Assume that $c \in \ZZ^n$ is given, such that $\langle c, h_i \rangle > 0$, where $h_i$ are the columns of $A^* = \Delta \cdot A^{-1}$, for $i \in \intint n$. Denote $\psi = \max\limits_{i \in \intint n}\Bigl\{\abs{\langle c, h_i \rangle}\Bigr\}$. Let, additionally, $S = P A Q$ be the SNF of $A$, where $P,Q \in \ZZ^{n \times n}$ are unimodular, and denote $\sigma = S_{n n}$. 

Then, for any $\tau >0$, the series $\fG(\PC,c;\,\tau)$ converges absolutely to a function of the type
$$
\frac{\sum\limits_{i = -n \cdot \sigma\cdot \psi}^{n \cdot \sigma\cdot \psi} \epsilon_i \cdot e^{\alpha_i \cdot \tau}}{\bigl(1 - e^{-\beta_1 \cdot \tau}\bigr)\bigl(1 - e^{-\beta_2 \cdot \tau}\bigr) \dots \bigl(1 - e^{-\beta_n \cdot \tau}\bigr)}, 
$$ where $\epsilon_i \in \ZZ_{\geq 0}$, $\beta_i \in \ZZ_{>0}$, and $\alpha_i \in \ZZ$. This representation can be found with an algorithm, having the arithmetic complexity bound
$$
O\bigl(T_{\SNF}(n) + \Delta \cdot n^2 \cdot \sigma \cdot \psi \bigr),
$$ where $T_{SNF}(n)$ is the arithmetic complexity of computing the SNF for $n \times n$ integer matrices.
\end{lemma}
\begin{proof}

After the unimodular map $x = Q x'$ and introducing slack variables $y$, the system $\{x \in \ZZ^n \colon A x \leq b\}$ becomes 
$$
\begin{cases}
S x + P y = P b\\
x \in \ZZ^{n}\\
y \in \ZZ^{n}_{\geq 0}.
\end{cases}
$$
Since $P$ is unimodular, the last system is equivalent to the system
\begin{equation}\label{initial_group_system_eq}
\begin{cases}
P y = P b \pmod{S \ZZ^n}\\
y \in \ZZ^{n}_{\geq 0}.
\end{cases}
\end{equation}
Denoting $\GC = \ZZ^{n}/S\ZZ^n$, $g_0 = P b \bmod S\ZZ^n$, $g_i = P_{* i} \bmod S\ZZ^n$, the last system \eqref{initial_group_system_eq} can be rewritten:
\begin{equation}\label{group_system_eq}
\begin{cases}
\sum\limits_{i = 1}^n y_i g_i = g_0\\
y \in \ZZ_{\geq 0}^n.
\end{cases}    
\end{equation}
Note that points $x \in \PC \cap \ZZ^n$ and the solutions $y$ of the system \eqref{group_system_eq} are connected by the bijective map $x = A^{-1}(b - y)$. Let $r_i = \abs{\langle g_i \rangle}$, for $i \in \intint{n}$, and $r_{\max} := \max_{i \in \intint n} \{r_i\}$. Clearly, $\abs{\GC} = \abs{\det(S)} = \Delta$ and $r_{\max} \leq \sigma$. For $k \in \intint{n}$ and $g' \in \GC$, let $\MC_k(g')$ be the solutions set of the auxiliary system $$
\begin{cases}
\sum\limits_{i = 1}^k y_i g_i = g'\\
y \in \ZZ_{\geq 0}^k,
\end{cases}
$$ and define
$$
\gG_k(g'; \tau) = \sum\limits_{y \in \MC_k(g')} e^{-\langle c, \sum\limits_{i=1}^k h_{i} y_i 
\rangle \tau}
$$
\begin{multline}\label{group_connection_eq}
    \text{It follows that}\quad \fG(\PC,c;\,\tau) = \sum\limits_{z \in \PC \cap \ZZ^{n}} e^{\langle c, z 
 \rangle \tau}
    = \sum\limits_{y \in \MC_n(g_0)} e^{\langle c, A^{-1}(b-y) \rangle \tau} = \\
    = e^{ \langle c, A^{-1} b \rangle \tau} \cdot \sum\limits_{y \in \MC_n(g_0)} e^{-\frac{1}{\Delta} \langle c, A^* y \rangle \tau} =  e^{ \langle c, A^{-1} b \rangle \tau} \cdot \gG_n\bigl(g_0;\frac{\tau}{\Delta}\bigr).
\end{multline}
The formulae for $\gG_k(g';\tau)$ were formally proven in \cite[see its formulae (10), (11), and (12)]{Counting_FPT_Delta}, we cite them in the following separate lemma. Since the original published paper \cite{Counting_FPT_Delta} contained an inaccuracy in the main result, we give a self-contained proof of the lemma in Subsection \ref{g_formulae_proof} of Appendix.
\begin{lemma}\label{g_formulae_lm}
The following formulae hold:
    \begin{gather}
    \gG_1(g'; \tau) = \frac{e^{- \langle c, s  h_1 \rangle 
 \tau}}{1 - e^{- \langle c, r_1  h_1 \rangle \tau}},\quad\text{where $s = \min\{y_1 \in \ZZ_{\geq 0} \colon y_1 \cdot g_1 = g' \}$},\label{gg_k_tau_initial}\\
    \gG_k(g';\tau) = \frac{1}{1 - e^{-\langle c, r_k  h_k \rangle \tau}} \cdot \sum\limits_{i = 0}^{r_k-1} e^{- \langle c, i h_k \rangle \tau} \cdot \gG_{k-1}(g' - i \cdot g_k; \tau),\label{gg_k_tau_recur}\\
    \gG_k(g';\tau) = \frac{\sum\limits_{i = - k \cdot \sigma \cdot \psi}^{k \cdot \sigma \cdot \psi} \epsilon_i(k,g') \cdot e^{- i \tau}}{\bigl(1 - e^{-\langle c, r_1 \cdot h_1 \rangle \tau}\bigr)\bigl(1 - e^{-\langle c, r_2 h_2 \rangle \tau}\bigr) \dots \bigl(1 - e^{- \langle c, r_k h_k \rangle \tau}\bigr)},\label{gg_k_tau_conv}
\end{gather}
where $\epsilon_{i}(k,g') \in \ZZ_{\geq 0}$ are coefficients, depending on $k$ and $g'$. If the set $\{y_1 \in \ZZ_{\geq 0} \colon y_1 g_1 = g' \}$ is empty, we put $\gG_1(g'; \tau) := 0$. If the vector $c$ is chosen such that $\langle c, h_i \rangle > 0$, for all $i \in \intint n$, then, for any $\tau > 0$, $k \in \intint n$, and $g' \in \GC$, the series $\gG_k(g'; \tau)$ converges absolutely to the corresponding r.h.s.\, functions.
\end{lemma} 
Let us estimate the complexity to compute the representation \eqref{gg_k_tau_conv} of $\gG_k(g';\tau)$, for all $k \in \intint n$ and $g' \in \GC$, using the recurrence \eqref{gg_k_tau_recur}. In comparison to the paper \cite{Counting_FPT_Delta}, we will use a bit more sophisticated and efficient algorithm to do that.  Consider a quotient group $\QS_k = \GC/\langle g_k \rangle$ and fix $\QC \in \QS_k$. Clearly, $\QC = q + \langle g_k \rangle$, where $q \in \GC$ is a member of $\QC$, and $r_k = \abs{\QC}$. For $j \in \intint[0]{r_k-1}$, define
\begin{equation}\label{hh_def}
\hG_k(j;\tau) = \bigl(1 - e^{-\langle c, r_1 h_1 \rangle \tau}\bigr) \cdot \ldots \cdot \bigl(1 - e^{- \langle c, r_k h_k \rangle \tau}\bigr) \cdot \gG_k(q + j \cdot g_k;\tau).    
\end{equation}
For the sake of simplicity, denote $x \ominus_{k} y = (x - y) \bmod r_k$, then the formulas \eqref{gg_k_tau_initial}, \eqref{gg_k_tau_recur} and \eqref{gg_k_tau_conv} become
\begin{gather}
     \hG_1(j; \tau) = e^{- \langle c, s  h_1 \rangle \tau},\quad\text{where $s = \min\{y_1 \in \ZZ_{\geq 0} \colon y_1 g_1 = q + j \cdot g_1 \}$},\label{hh_k_tau_initial}\\
    \hG_k(j;\tau) = \sum\limits_{i = 0}^{r_k-1} e^{- \langle c, i h_k \rangle \tau} \cdot \hG_{k-1}\bigl(j \ominus_{k} i; \tau\bigr),\label{hh_k_tau_recur}\\
    \hG_k(j;\tau) = \sum\limits_{i = - k \cdot \sigma \cdot \psi}^{k \cdot \sigma \cdot \psi} \epsilon_i(k,q + j \cdot g_k) \cdot e^{- i \tau}.\label{hh_k_tau_conv}
\end{gather}
Assume first that $k=1$. Then, clearly, all the values $$\hG_1(0;\tau), \hG_1(1;\tau), \dots, \hG_1(r_1-1;\tau)$$ can be computed with $O(r_1)$ operations.
Assume now that $k \geq 2$ and that $(k-1)$-th level has already been computed. By the $k$-th level, we mean all the functions $\hG_k(j;\tau)$, for $j \in \intint[0]{r_k-1}$.
Due to the formula \eqref{hh_k_tau_conv}, $\hG_k(j;\tau)$ contains $O(k \cdot \sigma \cdot \psi)$ monomials. Hence, the function $\hG_k(0;\tau)$ can be computed directly using the formula \eqref{hh_k_tau_recur} with $O(r_k \cdot k \cdot \sigma \cdot \psi)$ operations. For $j \geq 1$, we have
\begin{multline}\label{hh_k_tau_smart_recur}
    \hG_k(j;\tau) = \sum\limits_{i = 0}^{r_k-1} e^{- \langle c, i h_k \rangle \tau} \cdot \hG_{k-1}(j \ominus_k i; \tau) =\\
    = \sum\limits_{i = -1}^{r_k-2} e^{- \langle c, (i+1) h_k \rangle \tau} \cdot \hG_{k-1}\bigl( j \ominus_k (i+1) ; \tau\bigr) =\\
    = e^{- \langle c, h_k \rangle \tau} \cdot \hG_k(j-1;\tau) +\hG_{k-1}(j;\tau) - e^{- \langle c, r_k h_k \rangle \tau} \cdot \hG_{k-1}\bigl(j \ominus_k r_k;\tau\bigr) = \\
    = e^{- \langle c, h_k \rangle \tau} \cdot \hG_k(j-1;\tau) + (1 - e^{- \langle c, r_k h_k \rangle \tau}) \cdot \hG_{k-1}\bigl(j;\tau\bigr).
\end{multline}
Consequently, in the assumption that the $(k-1)$-th level has already been computed and that $h_k(0;\tau)$ is known, all the functions $h_k(1;\tau), \dots, h_k(r_k-1;\tau)$ can be computed with $O(r_k \cdot k \cdot \sigma \cdot \psi)$ operations, using the last formula \eqref{hh_k_tau_smart_recur}. 

In turn, when the functions $\hG_k(j;\tau)$, for $j \in \intint[0]{r_k-1}$, are already computed, we can return to the functions $\gG_k(g';\tau)$, for $g' = q + j \cdot g_k$, using the formula \eqref{hh_def}. It will consume additional $O(r_k)$ group operations to compute $g' = q + j \cdot g_k$. By the definition of $\GC$, the arithmetic cost of a single group operation can be estimated by the number of elements on the diagonal of $S$ that are not equal to $1$. Clearly, this number is bounded by $\min\{n, \log_2(\Delta)\}$. Consequently, the arithmetic cost of the last step is $O(r_k \cdot n)$, which is negligible in comparison with the $\hG_k(j;\tau)$ computational cost. 

Summarizing, we need $O(r_k \cdot k \cdot \sigma \cdot \psi)$ operations to compute $\gG_k(g';\tau)$, for $g' = q + j \cdot g_k$ and $j \in \intint[0]{r_k}$. Therefore, since $\abs{\QS} = \Delta/r_k$, the arithmetic computational cost to compute $k$-th level of $\gG_k(\cdot)$ is $$
O(\Delta \cdot k \cdot \sigma \cdot \psi),
$$ and the total arithmetic cost to compute all the levels is $$
O(\Delta \cdot n^2 \cdot \sigma \cdot \psi).
$$
Finally, using the formula \eqref{group_connection_eq}, we construct the function
\begin{multline*}
    \fG(\PC,c;\tau) = e^{\langle c, A^{-1} b \rangle \tau} \cdot \gG_n\bigl(g_0; \frac{\tau}{\Delta}\bigr) = \\
    = \frac{\sum\limits_{i = - k \cdot \sigma \cdot \psi}^{k \cdot \sigma \cdot \psi} \epsilon_i \cdot e^{\frac{1}{\Delta}\bigl(\langle c, A^* b \rangle - i\bigr) \tau}}{\bigl(1 - e^{-\langle c, \frac{r_1}{\Delta} h_1 \rangle \tau}\bigr)\bigl(1 - e^{-\langle c, \frac{r_2}{\Delta} h_2 \rangle \tau}\bigr) \dots \bigl(1 - e^{- \langle c, \frac{r_n}{\Delta} h_n \rangle \tau}\bigr)},
\end{multline*}
where $\epsilon_i:= \epsilon_i(n,g_0)$, which gives the desired representation of $\fG(\PC,c;\tau)$. Since $\gG_n(g_0;\tau)$ converges absolutely, for all $\tau >0$, the same is true for $\fG(\PC,c;\tau)$. Clearly, the arithmetic cost of the last transformation is proportional to the nominator length of $\gG_n(g_0;\tau)$, which is $O(n \cdot \sigma \cdot \psi)$.
\end{proof}

It is known  that a slight perturbation in the right-hand side of a system $A x \leq b$ can transform the polyhedron $\PC(A,b)$ to a simple one. We refer to the work \cite{epsilon_perturb} of Megiddo \& Chandrasekaran.
For $\varepsilon \in (0,1)$ and $i \in \intint m$, denote $t_{\varepsilon} \in \QQ^m$ to be a vector with $(t_{\varepsilon})_i = \varepsilon^i$.
\begin{theorem}[Megiddo \& Chandrasekaran \cite{epsilon_perturb}]\label{poly_simplification_th}
For any input matrix $A \in \ZZ^{m \times n}$ with $\rank(A) = n$, there exists a rational value $\varepsilon_A \in (0,1)$, such that, for any $b \in \ZZ^m$ and any $\varepsilon \in (0, \varepsilon_A]$, the polyhedron $\PC(A,b + t_{\varepsilon})$ is simple. 

The value $\varepsilon_A$ can be computed by a polynomial-time algorithm. More precisely, the algorithm needs $O(\log n)$ operations with numbers of size $O\bigl( n \cdot \log\bigl(n \norm{A}_{\max}\bigr) \bigr)$.
\end{theorem}

\begin{remark}\label{poly_simplification_rm} Let us discuss how to apply Theorem \ref{poly_simplification_th} to systems with rational r.h.s. For $A \in \ZZ^{m \times n}$ with $\rank(A) = n$ and $b \in \QQ^m$, let $\PC := \PC(A,b)$ be an $n$-dimensional polyhedron. Let us show how to construct a vector $t \in \QQ^m$, such that the polyhedron $\PC(A, b + t)$ will be simple and integrally equivalent to $\PC$.

To this end, let $D \in \ZZ_{\geq 0}^{m \times m}$ be the diagonal matrix, composed of the denominators of the corresponding components of $b$. Note that $\PC = \PC(D A, D b)$. Next, we apply Theorem \ref{poly_simplification_th} to the matrix $D A$, and let $\varepsilon$ be the resulting perturbation value. Since $D b$ is an integer and $0 < \varepsilon < 1$, the polyhedron $\PC(A, b + D^{-1} t_{\varepsilon}) = \PC(D A, D b + t_{\varepsilon})$ is simple and integrally equivalent to $\PC$. Consequently, we can put $t := D^{-1} t_{\varepsilon}$. Additionally, note that the described procedure needs only $O(m)$ operations to calculate $t$. 
\end{remark}


\subsection{The Proof of Theorem \ref{FasterCounting_th}}\label{main_th_0_proof}
\begin{proof}
Since any system in the standard form can be straightforwardly transformed to a system in the canonical form without changing the solutions set, assume that the polytope $\PC$ is defined by a system $A x \leq b$, where $A \in \ZZ^{m \times n}$ and $b \in \QQ^m$. 

Since $\PC$ is bounded, it follows that $\rank(A) = n$. Since $b$ is a rational vector, we can assume that $\gcd(A_j) = 1$, for all $j \in \intint m$. Now, let us assume that $\dim(\PC) < n$. Clearly, it is equivalent to the existence of an index $j \in \intint m$, such that $A_j x = b_j$, for all $x \in \PC$. Note that such $j$ could be found by a polynomial-time algorithm. W.l.o.g., assume that $j = 1$. Since $\gcd(A_1) = 1$, there exists a unimodular matrix $Q \in \ZZ^{n \times n}$ such that $A_1 = (1\,\BZero_{n-1}) Q$.
After the unimodular map $x' = Qx$, the system $A x \leq b$ transforms to the integrally equivalent\footnote{Saying "integrally equivalent" we mean that the sets of integer solutions of both systems are connected by a bijective unimodular map.} system 
$$
\begin{pmatrix}
    1 & \BZero_{n-1}\\
    h & B\\
\end{pmatrix} x \leq b,
$$ where $h \in \ZZ^{m-1}$ and $B \in \ZZ^{(m-1)\times(n-1)}$. Note that $\Delta(B) = \Delta(A) = \Delta$. Since the first inequality always holds as an equality on the solutions set, we can just substitute $x_1 = b_1$. As the result, we achieve a new integrally equivalent system with $n-1$ variables $B x \leq b'$, where $b' = b_{\intint[2]{m}} - b_1 \cdot h$. 

Due to the proposed reasoning, we can assume that $\dim(\PC) = n$. Let us make some more assumptions on $\PC$. Due to Theorem \ref{poly_simplification_th} and Remark \ref{poly_simplification_rm}, using $O(m)$ operations, we can produce a new r.h.s. vector $b' \in \QQ^m$, such that a new polytope, defined by $A x \leq b'$, will be simple and integrally equivalent to $\PC$. Consequently, we can assume that $\PC$ is simple. Let $v \in \vertex(\PC)$, denote
\begin{gather*}
    \JC(v) = \{j \colon A_j v = b_j\},\text{and}\\
    \PC_v = \{x \in \RR^n\colon A_{\JC(v)} x \leq b_{\JC(v)}\}.
\end{gather*}
Since $\PC$ is simple, it follows that $A_{\JC(v)} \in \ZZ^{n \times n}$ and $0 < \det(A_{\JC(v)}) \leq \Delta$. Due to the seminal work \cite{AvisFukuda} due to Avis \& Fukuda, all vertices of the simple polyhedron $\PC$ can be enumerated with $O\bigl( m \cdot n \cdot \abs{\vertex(\PC)} \bigr)$ arithmetic operations. Due to Lee, Paat, Stallknecht \& Xu \cite{ModularDiffColumns}, a $\Delta$-modular system has at most $O(n^2 \cdot \Delta^2)$ inequalities, i.e. $m = O(n^2 \cdot \Delta^2)$. Hence, all the polyhedra $\PC_v$ can be constructed with $O(\nu \cdot n^3 \cdot \Delta^2)$ operations. 

Define the set $\EC$ of \emph{edge directions} by the following way:
$$
h \in \EC \;\Longleftrightarrow\; h \text{ is a column of } -A^{*}_{\JC(v)} \text{ for some $v \in \vertex(\PC)$},
$$ where $B^* = \abs{\det(B)} \cdot B^{-1}$, for arbitrary invertible $B$. Assume that a vector $c \in \ZZ^n$ is chosen, such that $c^\top h \not= 0$, for each $h \in \EC$, and denote $\psi = \max\limits_{h \in \EC} \bigl\{\abs{c^\top h}\bigr\}$. Note that such a choice of the vector $c$ satisfies the conditions of Lemma \ref{quad_system_lem} applied to any polyhedra $\PC_v$, for $v \in \vertex(\PC)$. We use Lemma \ref{quad_system_lem} to all $\PC_v$ with the proposed choice of $c$, and construct the corresponding functions $f_v(\tau)$. Since $\sigma \leq \Delta$, and, due to Storjohann \cite{SNFOptAlg}, $T_{SNF}(n) = O(n^3)$, the arithmetic complexity of the last operation can be estimated by 
\begin{equation}\label{eta_full_complexity_eq}
    O(\nu \cdot \psi \cdot n^2 \cdot \Delta^2).
\end{equation}
Denote, additionally, $f_{\PC}(\tau) = \sum_{v \in \vertex(\PC)} f_v(\tau)$. Due to Brion's theorem \cite{Brion} (see also \cite[Chapter~6]{BarvBook}), we have
\begin{equation}\label{Brions_decomp_eq}
    [\PC] = \sum\limits_{v \in \vertex(\PC)} [\PC_v] \lmod\footnote{The words "modulo polyhedra with lines" mean that the sum can contain additional terms of the form $[\MC]$, where $\MC$ is a polyhedron with lines.}
\end{equation}
Consequently, it follows from Theorem \ref{Pukhlikov_th} and the last formula \eqref{Brions_decomp_eq} that, for any $\tau \in \RR$, the series $\fG(\PC,c; \tau)$ absolutely converges to the function $f_{\PC}(\tau)$. Therefore, to calculate $\abs{\PC \cap \ZZ^n}$, we need to compute $\lim\limits_{\tau \to 0} f_{\PC}(\tau)$. We follow to \cite[Chapter~14]{BarvBook}, to compute $\abs{\PC \cap \ZZ^n} = \lim\limits_{\tau \to 0} f_{\PC}(\tau)$ as a constant term in the Taylor decomposition of $f_{\PC}(\tau)$. Clearly, the constant term of $f_{\PC}(\tau)$ is just the sum of constant terms of $f_v(\tau)$, for $v \in \vertex(\PC)$. By this reason, let us fix some $v$ and consider
$$
f_v(\tau) = \frac{\sum\limits_{i = -n \cdot \sigma\cdot \psi}^{n \cdot \sigma\cdot \psi} \epsilon_i \cdot e^{\alpha_i \cdot \tau}}{\bigl(1 - e^{-\beta_1 \cdot \tau}\bigr)\bigl(1 - e^{-\beta_2 \cdot \tau}\bigr) \dots \bigl(1 - e^{-\beta_n \cdot \tau}\bigr)}, 
$$ where $\epsilon_i \in \ZZ_{\geq 0}$, $\beta_i \in \ZZ_{>0}$ and $\alpha_i \in \ZZ$. Due to \cite[Chapter~14]{BarvBook}, we can see that the constant term in the Taylor decomposition for $f_v(\tau)$ is exactly
\begin{equation}\label{constant_Taylor}
\sum\limits_{i = - n \cdot \sigma\cdot \psi}^{n \cdot \sigma \cdot \psi} \frac{\epsilon_i}{\beta_1 \dots \beta_n} \sum\limits_{j = 0}^n \frac{\alpha_i^j}{j!} \cdot \toddp_{n-j}(\beta_1, \dots, \beta_n),
\end{equation}
where $\toddp_j(\beta_1,\dots,\beta_n)$ is a homogeneous polynomial of degree $j$, called the \emph{$j$-th Todd polynomial} on $\beta_1,\dots,\beta_n$. Due to \cite[Theorem~7.2.8, p.~137]{AlgebracILP}, the values of $\toddp_{j}(\beta_1,\dots,\beta_n)$, for $j \in \intint n$, can be computed with an algorithm that is polynomial in $n$ and the bit-encoding length of $\beta_1,\dots,\beta_n$. Moreover, it follows from the theorem's proof that the arithmetic complexity can be bounded by $O(n^3)$. Therefore, it is not hard to see that we need $O(n^3 + n^2 \cdot \sigma \cdot \psi)$ operations to compute the value of \eqref{constant_Taylor}, and the total arithmetic cost to find the constant term in the Taylor's decomposition of the whole function $f_{\PC}(\tau)$ is $O\bigl(\nu \cdot ( n^3 + n^2 \cdot \sigma \cdot \psi) \bigr)$. 
Let us make an assumption that $\psi$ can be upper bounded by a function that grows as $\Omega(n)$. In this assumption, the complexity bound $O\bigl(\nu \cdot ( n^3 + n^2 \cdot \sigma \cdot \psi) \bigr)$ is negligible with respect to \eqref{eta_full_complexity_eq}. Hence, we can assume that the formula \eqref{eta_full_complexity_eq} bounds the arithmetic complexity of the algorithm at the current state.

Previously, we made the assumption that the vector $c \in \ZZ^n$ is chosen such that $c^\top h \not= 0$, for any $h \in \EC$. Let us present an algorithm that generates a vector $c$ with a respectively small value of the parameter $\psi = \max\limits_{h \in \EC} \bigl\{\abs{c^\top h}\bigr\}$. 
The main idea is concentrated in following Theorem \ref{all_non_zero_th} due to corrected version \cite{Counting_FPT_Delta_corrected} of the paper \cite{Counting_FPT_Delta}. Since at the current moment of time the corrections \cite{Counting_FPT_Delta_corrected} are available only as a preprint, we give a self-contained proof of Theorem \ref{all_non_zero_th} in Subsection \ref{all_non_zero_proof} of Appendix.
\begin{theorem}[Theorem~2 of \cite{Counting_FPT_Delta}]\label{all_non_zero_th}
Let $\AC$ be a set composed of $m$ non-zero vectors in $\QQ^n$. Then, there exists a randomized algorithm with the expected arithmetic complexity $O(n \cdot m)$, which finds a vector $z \in \ZZ^n$ such that:
\begin{enumerate}
    \item $a^\top z \not= 0$, for any $a \in \AC$;
    \item $\|z\|_{\infty} \leq m$.
\end{enumerate}
\end{theorem}

Since the polytope $\PC$ is assumed to be simple, each vertex $v \in \vertex(\PC)$ corresponds to exactly $n$ edge directions. Consequently, $2 \cdot \abs{\EC} = \nu \cdot n$. Choose some basis sub-matrix $B$ of $A$. Note that $B h \not= \BZero$ and $(B h)_i \in \intint[-\Delta]{\Delta}$, for any $h \in \EC$ and $i \in \intint n$. Next, we use Theorem \ref{all_non_zero_th} to the set $B \cdot \EC$, which produces a vector $z$, such that 
\begin{enumerate}
    \item $z^\top B h \not= 0$, for each $h \in \EC$;
    \item $\|z\|_{\infty} \leq \nu \cdot n$.
\end{enumerate}
Now, we assign $c := B^\top z$. By the construction, we have $c^\top h \not= 0$ and $\abs{c^\top h} = \abs{z^\top B h} \leq n^2 \cdot \nu \cdot \Delta$, for each $h \in \EC$. Therefore, $\psi \leq n^2 \cdot \nu \cdot \Delta$, which justifies the assumption on $\psi$. Due to the formula \eqref{eta_full_complexity_eq}, the total complexity bound becomes $O(\nu^2 \cdot n^4 \cdot \Delta^3)$, which finishes the proof.

\begin{remark}\label{inaccuracy_rm}
    Let us discuss an inaccuracy of the paper \cite{Counting_FPT_Delta}, which was corrected in the preprint \cite{Counting_FPT_Delta_corrected}. It has just been proven that there exists a vector $c \in \ZZ^n$, such that $c^\top h \not= 0$, for each $h \in \EC$, with $\psi \leq n^2 \cdot \nu \cdot \Delta$, where $\psi = \max\limits_{h \in \EC} \bigl\{\abs{c^\top h}\bigr\}$. In turn, the paper \cite{Counting_FPT_Delta} chooses the vector $c \in \ZZ^n$ by a different way that causes an error. More precisely, let $B$ be a basis sub-matrix of $A$, corresponding to some vertex $v \in \vertex(\PC)$. Since $\PC$ is assumed to be simple, $B$ is an $n \times n$ non-degenerate integer matrix. Then, the vector $c$ is chosen as the sum of columns of $B^\top$. It is easy to see that $\psi \leq n \cdot \Delta$, but the statement $\forall h \in \EC \colon c^\top h \not= 0$ is not necessary to be correct for every $\PC$, which is the mentioned inaccuracy.
\end{remark}
\end{proof}

\subsection{A bound for the number of vertices of a rational polyhedron}

For an arbitrary matrix $B \in \RR^{m \times n}$, denote 
$
\cone(B) = \{ Bt \colon t \in \RR^n_{\geq 0}\}
$.
The following lemmas help to estimate the number of vertices in a polyhedron, defined by a sparse system. We will use this bound to prove Theorem \ref{main_th_1}.
\begin{lemma}\label{ball_lm}
Let $A \in \ZZ^{n \times n}$, $\det(A) \not= 0$, and $\|\cdot\| \colon \RR^n \to \RR_{\geq 0}$ be any vector norm, which is symmetric with respect to any coordinate, i.e. $\|x\| = \|x - 2 x_i \cdot e_i\|$, for any $x \in \RR^n$ and $i \in \intint n$.  Let us consider a sector $\UC = \BB_{\|\cdot\|} \cap \cone(A)$, where $\BB_{\|\cdot\|} = \{x \in \RR^n \colon \|x\| \leq 1\}$ is the unit ball with respect to the $\|\cdot\|$-norm. Then,
\begin{equation}\label{ball_ineq}
    \vol(\UC) \geq \frac{\abs{\det(A)}}{2^n} \cdot \vol(r \cdot \BB_{\|\cdot\|}),
\end{equation} where $r \cdot \BB_{\|\cdot\|}$ is the $\|\cdot\|$-ball of the maximum radius $r$, inscribed into the set $\{x \in \RR^n \colon \|A x\| \leq 1\}$.

Consequently, let $\UC_1 = \BB_1 \cap \cone(A)$ and $\UC_{\infty} = \BB_{\infty} \cap \cone(A)$. Then,
\begin{gather}
    \vol(\UC_1) \geq \frac{\abs{\det(A)}}{(2 \|A\|_\infty)^n} \cdot \vol(\BB_1) \geq \frac{\abs{\det(A)}}{\bigl(2 \|A\|_{\max} \cdot \colSparse(A)\bigr)^n} \cdot \vol(\BB_1); \label{l1_cone_ineq}\\
    \vol(\UC_\infty) \geq \frac{\abs{\det(A)}}{(2 \|A\|_1)^n} \cdot \vol(\BB_\infty) \geq \frac{\abs{\det(A)}}{\bigl(2 \|A\|_{\max} \cdot \rowSparse(A)\bigr)^n} \cdot \vol(\BB_\infty).\label{linf_cone_ineq}
\end{gather}
\end{lemma}
\begin{proof}
Let us prove the inequality \eqref{ball_ineq}. Clearly,
$$
\vol(\UC) = \abs{\det(A)} \cdot \vol\bigl( \KC \cap \cone(I_{n \times n}) \bigr),
$$ where $\KC = \{x \in \RR^n \colon \|A x\| \leq 1\}$. By the definition of $r$, we have 
$
\KC \supseteq r \cdot \BB_{\|\cdot\|}
$. Consequently,
$$
\vol(\UC) \geq \abs{\det(A)} \cdot \vol\bigl( r \cdot \BB_{\|\cdot\|} \cap \cone(I_{n \times n}) \bigr) \geq \frac{\abs{\det(A)}}{2^n} \cdot \vol(r \cdot \BB_{\|\cdot\|}).
$$
Now, let us prove the inequality \eqref{l1_cone_ineq}. To this end, we just need to prove the inequality $r \geq \frac{1}{\|A\|_\infty}$ with respect to the $l_1$-norm. Let us consider the set $\KC$. It can be represented as the set of solutions of the following inequality:
\begin{equation}\label{linf_KC_repres}
    \sum\limits_{i = 1}^n \abs{A_{i *} x} \leq 1.
\end{equation}
Let us consider the $2 n$ points $\pm p_i = \pm \frac{1}{\|A\|_{\infty}} \cdot e_i$, for $i \in \intint n$. Substituting $\pm p_j$ to the inequality \eqref{linf_KC_repres}, we have 
$$
\sum\limits_{i = 1}^n \abs{A_{i *} p_j} = \frac{1}{\|A\|_{\infty}} \cdot \sum\limits_{i = 1}^n \abs{A_{i *} e_j} = \frac{1}{\|A\|_{\infty}} \cdot \sum\limits_{i = 1}^n \abs{A_{i j}} \leq 1.
$$
Hence, all the points $\pm p_i$, for $i \in \intint n$, satisfy the inequality \eqref{linf_KC_repres}. Since $\KC$ is convex, we have $\frac{1}{\|A\|_\infty} \cdot \BB_1 \subseteq \KC$, and, consequently, $r \geq \frac{1}{\|A\|_\infty}$.

Finally, let us prove the inequality \eqref{linf_cone_ineq}. Again, we need to show that $r \geq \frac{1}{\|A\|_1}$ with respect to the $l_1$-norm. In the current case, the set $\KC$ can be represented as the set of solutions of the following system:
\begin{equation}\label{l1_KC_repres}
    \forall i \in \intint n, \quad \abs{A_{i *} x} \leq 1.
\end{equation}
Let us consider the set $\MC = \{\frac{1}{\|A\|_1} \cdot (\pm 1, \pm 1, \dots, \pm 1)^\top \}$ of $2^n$ points. Substituting any point $p \in \MC$ to the $j$-th inequality of the system \eqref{l1_KC_repres}, we have
$$
\abs{A_{j *} p} \leq \sum\limits_{i = 1}^n \abs{A_{j i}} \abs{p_i} = \frac{1}{\|A\|_1} \cdot \sum\limits_{i = 1}^n \abs{A_{j i}} \leq 1.
$$
Hence, all the points $p \in \MC$ satisfy the inequality \eqref{l1_KC_repres}. Since $\KC$ is convex, we have $\frac{1}{\|A\|_1} \cdot \BB_\infty \subseteq \KC$, and, consequently, $r \geq \frac{1}{\|A\|_1}$.
\end{proof}

\begin{lemma}\label{vertex_num_lm}
Let $A \in \ZZ^{m\times n}$, $b \in \QQ^{m}$, and $\rank(A) = n$. Let $\PC$ be a polyhedron, defined by a system $A x \leq b$. Then, 
$
\abs{\vertex(\PC)} \leq 2^n \cdot {\totn(A)}^n \leq \bigl(2 \|A\|_{\max}\bigr)^n \cdot \wSparse(A)^n
$.
\end{lemma}
\begin{proof}
Let $\NC(v) = \cone\bigl(A^\top_{\JC(v)}\bigr)$ be the normal cone of a vertex $v \in \vertex(\PC)$, where $\JC(v) = \bigl\{ j \in \intint m \colon A_{j *} v = b_j\bigr\}$. Since $\rank(A) = n$, we have $\dim\bigl( \NC(v) \bigr) = n$, for any $v \in \vertex(\PC)$. It is a known fact that $\dim\bigl(\NC(v_1)\cap\NC(v_2)\bigr) < n$, for different $v_1,v_2 \in \vertex(\PC)$. Next, we will use the following trivial inclusion
\begin{equation}\label{norm_cones_BB_eq}
  \bigcup\limits_{v \in \vertex(\PC)} \NC(v) \cap \BB \subseteq \BB,  
\end{equation} where $\BB$ is the unit ball with respect to any vector norm $\|\cdot\| \colon \RR^n \to \RR_{\geq 0}$.

Again, since $\rank(A) = n$, each matrix $A^\top_{\JC(v)}$ contains a non-degenerate $n \times n$ sub-matrix. Taking $\BB := \BB_1$ or $\BB := \BB_\infty$, by Lemma \ref{ball_lm}, we have $\vol(\NC(v) \cap \BB) \geq \frac{\vol(\BB)}{\bigl(2 \totn(A)\bigr)^n}$. Finally, due to \eqref{norm_cones_BB_eq}, we have
$$
\frac{\vol(\BB)}{\bigl(2 \totn(A)\bigr)^n} \cdot \abs{\vertex(\PC)} \leq \vol(\BB).
$$
\end{proof}

\subsection{The Proof of Theorem \ref{main_th_1}}\label{main_th_1_proof}

\begin{proof} 
Consider first the case, when $\PC$ is unbounded. In this case, we need only to distinguish between two possibilities: $\abs{\PC \cap \ZZ^n} = 0$ and $\abs{\PC \cap \ZZ^n} = +\infty$. Due to \cite[Theorem~17.1]{Schrijver}, if $\abs{\PC \cap \ZZ^n} \not= 0$, then there exists 
 $v \in \PC \cap \ZZ^n$ such that $\|v\|_{\infty} \leq (n+1) \cdot \Delta_{ext}$, where $\Delta_{ext} = \Delta(A_{ext})$ and $A_{ext} = \bigl(A\,b\bigr)$ is the extended matrix of the system $A x \leq b$. Consequently, to transform the unbounded case to the bounded one, we just need to add the inequalities $\abs{x_i} \leq (n+1) \cdot n^{n/2} \cdot \bigl(\|A_{ext}\|_{\max}\bigr)^n$, for $i \in \intint n$, to the original system $A x \leq b$.
 

Now, we can assume that $\PC$ is bounded, and consequently $\rank(A) = n$. Due to Theorem \ref{FasterCounting_th}, the counting problem can be solved by an algorithm with the arithmetic complexity bound 
\begin{equation}\label{counting_complexity_eq}
 O(\nu^2 \cdot n^4 \cdot \Delta^3 ),
\end{equation}
where $\nu$ is the maximum number of vertices in polyhedra with fixed $A$ and varying $b$. In our case, the value of $\nu$ can be estimated by Lemma \ref{vertex_num_lm}. To estimate the value of $\Delta$, we use the inequalities \eqref{hadamard_sparse_eq} and \eqref{hadamrd_norm_eq}. The inequalities for $\nu$ and $\Delta$, together with the bound \eqref{counting_complexity_eq}, give the desired complexity bound for the problem \ref{CProb}.  

Let us show how to find some point $z$ inside $\PC \cap \ZZ^n$ in the case $\abs{\PC \cap \ZZ^n} > 0$, to handle the problem \ref{FProb}. For $\alpha, \beta \in \ZZ$, let us consider the polytope $\PC'(\alpha, \beta)$, defined by the system $A x \leq b$ with the additional inequality $\alpha \leq x_1 \leq \beta$. The maximum rank-order sub-determinants of the new system are bounded by $\max\{\Delta_{n}, \Delta_{n-1}\}$. In turn, the value of $\Delta_{n-1}$ can be estimated in the same way, as it was done for $\Delta_n$. Let $v$ be some vertex of $\PC$, which can be found by a polynomial-time algorithm. Due to the seminal sensitivity result \cite{Sensitivity_Tardos} of Cook, Gerards, Schrijver \& Tardos, if $\PC \cap \ZZ^n \not= \emptyset$, then there exists a point $z \in \PC \cap \ZZ^n$ such that $\|v-z\|_{\infty} \leq n \cdot \Delta_{tot}$. So, the value of $z_1$ can be found, using the binary search with questions to the $\PC'(\alpha, \beta) \cap \ZZ^n$-feasibility oracle, which can be clearly reduced to the \ref{CProb} problem. Clearly, we need $O(\log(n \Delta_{tot}))$ calls to the oracle. After the moment, when we already know the value of $z_1$, we just add the equality $x_1 = z_1$ to the system $A x \leq b$ and start a similar search procedure for the value of $z_2$. The total number of calls to the binary search oracle to compute all the components of $z$ is $O(n \cdot \log(n \Delta_{tot}))$.

Finally, let us explain how to deal with the problem \ref{OptCProb}. Let $\alpha, \beta \in \ZZ$, consider the polytope $\PC'(\alpha, \beta)$, defined by the system $A x \leq b$ with the additional inequality $\alpha \leq c^\top x \leq \beta$. 
Let $A' \in \ZZ^{(m+2) \times n}$ be the matrix that defines $\PC'(\alpha, \beta)$, i.e. $A' = (c \; -c\; A^\top)$. Expanding sub-determinants of $A'$ by the $c^\top$-row, we have $\Delta(A') \leq \|c\|_{1} \cdot \Delta_{n-1}(A)$. 
Let us estimate the number of vertices in $\PC'(\alpha, \beta)$. The polytope $\PC'(\alpha, \beta)$ is the intersection of the polytope $\PC$ with the slab $\{x \in \RR^n \colon \alpha \leq c^\top x \leq \beta \}$. Clearly, the new vertices may appear only on edges of $\PC$, by at most $2$ new vertices per edge. The number of edges in $\PC$ is bounded by $\abs{\vertex(\PC)}^2 / 4$. In turn, the value of $\abs{\vertex(\PC)}^2 / 4$ can be estimated, using Lemma \ref{vertex_num_lm}.
Due to Theorem \ref{FasterCounting_th}, the value $\abs{\PC'(\alpha, \beta) \cap \ZZ^n}$ can be computed by an algorithm with the desired complexity bounds. To complete the proof, we note that, using the binary search method, the original optimization problem can be reduced to a polynomial number of feasibility questions in the set $\PC'(\alpha, \beta) \cap \ZZ^n$ for different $\alpha,\beta$.
\end{proof}

\backmatter

\addcontentsline{toc}{section}{Acknowledgement}
\section*{Acknowledgement}
Section \ref{main_result_sec} was prepared within the framework of the Basic Research Program at the National Research University Higher School of Economics (HSE). Subsection \ref{standard_ILP_subs} and Section \ref{comb_prob_sec} 
were prepared under financial support of Russian Science Foundation grant No 21-11-00194. Additionally, the authors would thank the anonymous reviewers who offered important comments that significantly improved our work.

\addcontentsline{toc}{section}{Statements and Declarations}
\section*{Statements and Declarations}

\noindent{\bf Competing Interests:} The authors have no competing interests.

\noindent{\bf Data availability statement:} The manuscript has no associated data.

\begin{appendices}

\section{Proof of Theorem \ref{all_non_zero_th}}\label{all_non_zero_proof}

\begin{proof}
Fix a parameter $r$ and let $\IC = \intint[-r]{r}$. For $a \in \AC$, denote $\HC_a = \{x \in \IC^n \colon a^\top x = 0\}$, and let 
$$
\NC = \IC^n \setminus \bigcup\limits_{a \in \AC} \HC_a.
$$
Consider a polynomial $f \colon \RR^n \to \RR$ given by the formula
$$
f(x) = \prod\limits_{a \in \AC} a^\top x.
$$
Clearly, $f$ is a homogeneous polynomial with $\deg(f) = \abs{\AC} = m$. Let $\RC = \{ x \in \IC^{n} \colon f(x) = 0\}$ be the roots of $f$ inside $\IC^{n}$. Note that $\RC = \bigcup_{a \in \AC} \HC_a$ and $\NC = \IC^n \setminus \RC$. Due to the known Schwartz--Zippel lemma, $\abs{\RC} \leq \deg(f) \cdot \abs{\IC}^{n-1} = m \cdot (2r+1)^{n-1}$. Therefore, $\abs{\NC} \geq (2 r + 1)^n - m \cdot (2 r + 1)^{n-1} = (2 r + 1)^{n-1} \cdot (2 r + 1 - m)$, and consequently
\begin{equation*}
\frac{\abs{\NC}}{\abs{\IC^n}} \geq \frac{2r +1 - m}{ 2 r + 1} = 1 - \frac{m}{2 r + 1}.
\end{equation*}

Assign $r := m$. After that, the previous inequality becomes $\frac{\abs{\NC}}{\abs{\IC^n}} > 1/2$. Now, to find a vector $z$ that can satisfy the claims 
\begin{enumerate}
    \item $a^\top z \not= 0$, for any $a \in \AC$;
    \item $\|z\|_{\infty} \leq m$;
\end{enumerate} we uniformly sample points $z$ inside $\IC^n$. With a probability at least $1/2$ it will satisfy the first claim. The second claim is satisfied automatically. Therefore, the expected number of sampling iterations is $O(1)$. The arithmetic complexity of a single iteration is clearly bounded by $O(n \cdot m)$, which completes the proof.
\end{proof}

\section{Proof of the Lemma \ref{g_formulae_lm}}\label{g_formulae_proof}

\subsection{A Recurrent Formula for the Generating Function of a Group Polyhedron}

Let $\GC$ be an arbitrary finite Abelian group and $g_1,\dots,g_n \in \GC$. Let additionally $r_i = \abs{\langle g_i \rangle}$ be the order of $g_i$, for $i \in \intint n$, and $r_{\max} = \max_{i} \{r_i\}$. For $g' \in \GC$ and $k \in \intint n$, let $\MC(k,g')$ be the solutions set of the following system:
\begin{equation}\label{f_k_system}
    \begin{cases}
    \sum\limits_{i = 1}^k x_i g_i = g'\\
    x \in \ZZ_{\geq 0}^k.
    \end{cases}
\end{equation}
Consider the formal power series 
$
\fG_k(g';\xB) = \sum\limits_{z \in \MC(k,g') \cap \ZZ^k} \xB^z.
$
For $k = 1$, we clearly have
\begin{equation}\label{f_k_1form}
\fG_1(g';\xB) = \frac{x_1^s}{1 - x_1^{r_1}},\quad\text{where $s = \min\{x_1 \in \ZZ_{\geq 0} \colon x_1 g_1 = g'\}$.}    
\end{equation}
If such $s$ does not exist, we put $\fG_1(g';\xB) := 0$. Clearly, the series $\fG_1(g';\xB)$ absolutely converges to the corresponding r.h.s.\, function for any $x_1 \in \CC$ with $\abs{x_1^{r_1}} < 1$. For any value of $x_k \in \ZZ_{\geq 0}$, the system \eqref{f_k_system} can be rewritten as
\begin{equation*}
    \begin{cases}
    \sum\limits_{i = 1}^{k-1} x_i g_i = g' - x_k g_k\\
    x \in \ZZ_{\geq 0}^{k-1}.
    \end{cases}
\end{equation*}
Hence, for $k \geq 1$, we have
\begin{multline}\label{f_k_recurrence}
    \fG_k(g';\xB) = \\
    = \frac{ \fG_{k-1}(g';\xB) + x_{k} \cdot \fG_{k-1}(g' - g_k;\xB) + \dots + x_{k}^{r_k - 1} \cdot \fG_{k-1}(g' - g_k \cdot (r_k - 1);\xB)} {1 - x_k^{r_k}} = \\
    = \frac{1}{1 - x_{k}^{r_k}} \cdot \sum_{i = 0}^{r_k - 1} x_k^i \cdot \fG_{k-1}(g' - i \cdot g_k;\xB).
\end{multline}

\begin{equation}\label{f_k_conv}
\text{Consequently,} \quad \fG_k(g';\xB) = \frac{\sum\limits_{i_1 = 0}^{r_1-1}\dots\sum\limits_{i_k = 0}^{r_k-1} \epsilon_{i_1,\dots,i_k} x_1^{i_1} \dots x_k^{i_k}}{(1 - x_1^{r_1})(1 - x_2^{r_2})\dots(1 - x_k^{r_k})},    
\end{equation}
where the numerator is a polynomial with coefficients $\epsilon_{i_1,\dots,i_k} \in \{0,1\}$ and degree at most $(r_1 - 1) \dots (r_k - 1)$. Since a sum of absolutely convergent series is absolutely convergent, it follows from the induction principle that the series $\fG_k(g';\xB)$ absolutely converges to the r.h.s.\, of the formula \eqref{f_k_conv} when $\abs{x_i^{r_i}} < 1$ for each $i \in \intint k$.

\subsection{The group $\GC$, induced by the SNF, of $A$}
Recall that $A \in \ZZ^{n \times n}$, $0 < \Delta = \abs{\det(A)}$, and $h_1, \dots, h_n$ are the columns of $A^* := \Delta \cdot A^{-1}$. The vector $c \in \ZZ^n$ is chosen, such that $\langle c, h_i \rangle > 0$, for each $i \in \intint n$, and $\psi = \max_i \abs{ \langle c, h_i \rangle }$.  Additionally, let $S = P A Q$ be the SNF of $A$, where $P,Q \in \ZZ^{n \times n}$ are unimodular, and $\sigma = S_{n n}$.

Let us consider the sets $\MC(k,g')$, induced by the group system \eqref{f_k_system} with $\GC = \ZZ^{n}/S\ZZ^n$ and $g_i = P_{* i} \bmod S\ZZ^n$. Note that $r_i \leq \sigma$, for each $i \in \intint n$. Additionally, let us consider a new formal series, defined by 
$$
\hat \fG_k(g';\xB) = \sum\limits_{z \in \MC(k,g') \cap \ZZ^k} \xB^{-\sum\limits_{i=1}^k h_i z_i},
$$ which can be derived from the series $\fG_k(g';\xB)$ by the monomial substitution $x_i \to \xB^{-h_i}$.
For $\hat \fG_k(g';\xB)$, the formulae \eqref{f_k_1form}, \eqref{f_k_recurrence} and \eqref{f_k_conv} become:
\begin{gather}
    \hat \fG_1(g'; \xB) = \frac{\xB^{- s h_1}}{1 - \xB^{-r_1 h_1}},\quad\text{where $s = \min\{y_1 \in \ZZ_{\geq 0} \colon y_1 g_1 = g' \}$,}\label{ff_k_1form}\\
    \hat \fG_k(g';\xB) = \frac{1}{1 - \xB^{-r_k h_k}} \cdot \sum\limits\limits_{i = 0}^{r_k-1} \xB^{- i h_k} \cdot \hat \fG_{k-1}(g' - i \cdot g_k; \xB) \text{ and}\label{ff_k_recur}\\    
\hat \fG_k(g';\xB) = \frac{\sum\limits_{i_1 = 0}^{r_1-1}\dots\sum\limits_{i_k = 0}^{r_k-1} \epsilon_{i_1,\dots,i_k} \xB^{-(i_1 h_1 + \dots + i_k h_k)}}{(1 - \xB^{-r_1 h_1})(1 - \xB^{-r_2 h_2}) \dots (1 - \xB^{-r_k h_k})}\label{ff_k_conv}.
\end{gather}
Clearly, here the absolute convergence takes place for the values of $\xB$ with $\abs{\xB^{- r_i h_i}} < 1$, for each $i \in \intint k$. Let us consider now the formal series 
$$
\gG_k(g'; \tau) = \sum\limits_{y \in \MC_k(g')} e^{- \tau \cdot \langle c, \sum_{i=1}^k h_i y_i \rangle},
$$ which can be derived from $\hat \fG_k(g';\xB)$ by the substitution $x_i \to e^{\tau \cdot c_i}$. For $\gG_k(g'; \tau)$, the formulae \eqref{ff_k_1form}, \eqref{ff_k_recur}, and \eqref{ff_k_conv} become:
\begin{gather}
    \gG_1(g'; \tau) = \frac{e^{- \langle c, s  h_1 \rangle \cdot \tau}}{1 - e^{- \langle c, r_1  h_1 \rangle \cdot \tau}},\label{g_k_1form}\\
    \gG_k(g';\tau) = \frac{1}{1 - e^{-\langle c, r_k  h_k \rangle \cdot \tau}} \cdot \sum\limits_{i = 0}^{r_k-1} e^{- \langle c, i h_k \rangle \cdot \tau} \cdot \gG_{k-1}(g' - i \cdot g_k; \tau),\label{g_k_recur}\\
    \gG_k(g';\tau) = \frac{\sum\limits_{i_1 = 0}^{r_1-1}\dots\sum\limits_{i_k = 0}^{r_k-1} \epsilon_{i_1,\dots,i_k} e^{-\langle c, i_1 h_1 + \dots + i_k h_k \rangle \cdot \tau} }{\bigl(1 - e^{-\langle c, r_1 h_1 \rangle \cdot \tau}\bigr)\bigl(1 - e^{-\langle c, r_2 h_2 \rangle \cdot \tau}\bigr) \dots \bigl(1 - e^{- \langle c, r_k h_k \rangle \cdot \tau}\bigr)}.\label{g_k_conv}
\end{gather}
Since the series $\hat \fG_k(g';\xB)$ absolutely converges, when $\abs{\xB^{- r_i h_i}} < 1$, for each $i \in \intint k$, the new one converges, for any $\tau > 0$. Since $\langle c,h_i \rangle \in \ZZ_{\not=0}$, for each $i$, the number of terms $e^{-\langle c, \cdot \rangle \cdot \tau}$ is bounded by $2 \cdot k \cdot \sigma\cdot \psi + 1$. So, after combining similar terms, the numerator's length becomes $O(k \cdot \sigma\cdot \psi)$.
In other words, there exist coefficients $\epsilon_i \in \ZZ_{\geq 0}$, such that 
\begin{equation}\label{g_k_coef_form}
\gG_k(g';\tau) = \frac{\sum\limits_{i = - k \cdot \sigma \cdot \psi}^{k \cdot \sigma \cdot \psi} \epsilon_i \cdot e^{- i \cdot \tau}}{\bigl(1 - e^{-\langle c, r_1 \cdot h_1 \rangle \tau}\bigr)\bigl(1 - e^{-\langle c, r_2 h_2 \rangle \cdot \tau}\bigr) \dots \bigl(1 - e^{- \langle c, r_k h_k \rangle \cdot \tau}\bigr)}.    
\end{equation}
The formulae \eqref{g_k_1form}, \eqref{g_k_recur}, and \eqref{g_k_coef_form} coincide with the desired formulae \eqref{gg_k_tau_initial}, \eqref{gg_k_tau_recur}, and \eqref{gg_k_tau_conv}. So, the proof of Lemma \ref{g_formulae_lm} is finished.

\end{appendices}

\addcontentsline{toc}{section}{References}
\bibliography{grib_biblio}

\end{document}